\documentclass[11pt]{article}
\usepackage{amsfonts,amsmath,amssymb,amsthm}
\usepackage{algpseudocode}
\usepackage{algorithm,algorithmicx}
\usepackage{color}
\usepackage[pdftex]{graphicx}
\usepackage{comment}
\usepackage{hyperref}
\newlength{\defbaselineskip}
\newtheorem{theorem}{Theorem}
\newtheorem{corollary}{Corollary}
\newtheorem{definition}{Definition}
\newtheorem{lemma}{Lemma}

\DeclareMathOperator{\poly}{poly}
\DeclareMathOperator{\bigO}{\mathcal{O}}
\DeclareMathOperator{\nnz}{nnz}
\DeclareMathOperator{\A}{\mathcal{A}}
\DeclareMathOperator{\R}{\mathbb{R}}
\DeclareMathOperator{\T}{\mathcal{T}}
\DeclareMathOperator{\D}{\mathcal{D}}
\DeclareMathOperator{\E}{\mathcal{E}}

\allowdisplaybreaks

\begin{document}

\title{Low-distortion Subspace Embeddings in Input-sparsity Time \\ and
  Applications to Robust Linear Regression }

\author{
  Xiangrui Meng
  \thanks{
    Most of this work was done while the author was at 
    ICME, Stanford University supported by NSF DMS-1009005.
    Current affiliation: LinkedIn Corporation,
    Mountain View, 94403.
    Email: ximeng@linkedin.com.
  }
  \and
  Michael W. Mahoney
  \thanks{
    Dept.\ of Mathematics,
    Stanford University,
    Stanford, CA 94305.
    Email: mmahoney@cs.stanford.edu
  }
}

\date{}
\maketitle


\begin{abstract}%
\noindent
Low-distortion subspace embeddings are critical building blocks for developing
improved random sampling and random projection algorithms for common linear
algebra problems.
Here, we show that, given a matrix $A \in \R^{n \times d}$, with $n \gg d$, and
a $p \in [1, 2)$, with a constant probability, we can construct a low-distortion
embedding matrix $\Pi \in \R^{\poly(d) \times n}$ that embeds $\A_p$, the
$\ell_p$ subspace spanned by $A$'s columns, into $(\R^{\bigO(\poly(d))},
\|\cdot\|_p)$; the distortion of our embeddings is only $\bigO(\poly(d))$, and
we can compute $\Pi A$ in $\bigO(\nnz(A))$ time, i.e., input-sparsity time.
Our result generalizes the input-sparsity time $\ell_2$ subspace embedding
proposed recently by Clarkson and Woodruff; and for completeness, we present a
simpler and improved analysis of their construction for $\ell_2$.
These input-sparsity time $\ell_p$ embeddings are optimal, up to constants, in
terms of their running time; and the improved running time propagates to
applications such as $(1\pm \epsilon)$-distortion $\ell_p$ subspace embedding
and relative-error $\ell_p$ regression.
For $\ell_2$, we show that a $(1+\epsilon)$-approximate solution to the $\ell_2$
regression problem specified by the matrix $A$ and a vector $b \in \R^n$ can be
computed in $\bigO(\nnz(A) + d^3 \log(d/\epsilon) /\epsilon^2)$ time; and for
$\ell_p$, via a subspace-preserving sampling procedure, we show that a $(1\pm
\epsilon)$-distortion embedding of $\A_p$ into $\R^{\bigO(\poly(d))}$ can be
computed in $\bigO(\nnz(A) \cdot \log n)$ time, and we also show that a
$(1+\epsilon)$-approximate solution to the $\ell_p$ regression problem $\min_{x
  \in \R^d} \|A x - b\|_p$ can be computed in $\bigO(\nnz(A) \cdot \log n +
\poly(d) \log(1/\epsilon)/\epsilon^2)$ time.
Moreover, we can also improve the embedding dimension or equivalently the sample
size to $\bigO(d^{3+p/2} \log(1/\epsilon) / \epsilon^2)$ without increasing the
complexity.
\end{abstract}

\section{Introduction}
\label{sxn:intro}

Regression problems are ubiquitous, and the fast computation of their 
solutions is of interest in many large-scale data applications.
A parameterized family of regression problems that is of particular interest 
is the \emph{overconstrained $\ell_p$ regression problem}: given a matrix 
$A \in \R^{n \times d}$, with $n > d$, a vector $b \in \R^n$, a norm 
$\|\cdot\|_p$ parameterized by $p\in[1,\infty]$, and an error parameter 
$\epsilon > 0$, find a $(1+\epsilon)$-approximate solution 
$\hat{x} \in \R^d$ to:
\begin{equation}
  f^* = \min_{x \in \R^d} \|A x - b\|_p,
\end{equation}
i.e., find a vector $\hat{x}$ such that 
$\|A \hat{x} - b\|_p \leq (1+\epsilon) f^*$, where the $\ell_p$ norm of a
vector $x$ is $ \|x\|_p=\left(\sum_i|x_i|^p\right)^{1/p} $, defined to be
$\max_i |x_i|$ for $p=\infty$.
Special cases include the $\ell_2$ regression problem, also known as Least
Squares Approximation problem, and the $\ell_1$ regression problem, also known
as the Least Absolute Deviations or Least Absolute Errors problem.
The latter is of particular interest as a \emph{robust estimation} or
\emph{robust regression} technique, in that it is less sensitive to the presence
of outliers than the former.
We are most interested in this paper in the $\ell_1$ regression problem due 
to its robustness properties, but our methods hold for general $p\in[1,2]$, 
and thus we formulate our results in $\ell_p$.

It is well-known that for $p \ge 1$, the overconstrained $\ell_p$ regression
problem is a convex optimization problem; for $p=1$ and $p=\infty$, it is an
instance of linear programming; and for $p=2$, it can be solved with
eigenvector-based methods such as with the QR decomposition or the Singular
Value Decomposition of $A$.
In spite of their low-degree polynomial-time solvability, $\ell_p$ regression
problems have been the focus in recent years of a wide range of random sampling
and random projection algorithms, largely due to a desire to develop improved
algorithms for large-scale data applications~\cite{AMT10,MSM11_TR,CDMMMW13_SODA}.
For example, Clarkson~\cite{Cla05} uses subgradient and sampling methods to
compute an approximate solution to the overconstrained $\ell_1$ regression
problem in roughly $\bigO(nd^5\log n)$ time; and Dasgupta et
al.~\cite{DDHKM09_lp_SICOMP} use well-conditioned bases and subspace-preserving
sampling algorithms to solve general $\ell_p$ regression problems, for
$p\in[1,\infty)$, in roughly $\bigO(nd^5\log n)$ time.
A similar subspace-preserving sampling algorithm was developed by Drineas,
Mahoney, and Muthukrishnan~\cite{DMM06} to compute an approximate solution to
the $\ell_2$ regression problem.
The algorithm of~\cite{DMM06} relies on the estimation of the $\ell_2$ leverage
scores%
\footnote{Recall that for an $n \times d$ matrix $A$, with $n \gg d$, the
  \emph{$\ell_2$ leverage scores} of the rows of $A$ are equal to the diagonal
  elements of the projection matrix onto the span of $A$.
  That is, if $A=QR$ is a QR decomposition of $A$, or if $A=Q \Sigma V^T$ is the
  thin SVD of $A$, then the leverage scores equal the Euclidean norms squared of
  the rows of the $n \times d$ matrix $Q$, and thus they can be computed exactly
  in $\bigO(nd^2)$ time.
  See~\cite{Mah-mat-rev_BOOK,DMMW12_ICML} for details; and note that they can be
  generalized to $\ell_1$ and other $\ell_p$ norms~\cite{CDMMMW13_SODA} as well as
  to arbitrary $n \times d$ matrices, with both $n$ and $d$ large, if one
  specifies a low-rank parameter~\cite{CUR_PNAS,DMMW12_ICML}.}
of $A$ to be used as an importance sampling distribution, but when combined with
the results of Sarl\'{o}s~\cite{Sarlos06} and Drineas et
al.~\cite{DMMS07_FastL2_NM10} (that quickly preprocess $A$ to uniformize those
scores) or Drineas et al.~\cite{DMMW12_ICML} (that quickly computes
approximations to those scores), this leads to a random projection or random
sampling (respectively) algorithm for the $\ell_2$ regression problem that runs
in roughly $\bigO(n d \log d)$ time~\cite{DMMS07_FastL2_NM10,Mah-mat-rev_BOOK}.
More recently, Sohler and Woodruff~\cite{SW11} introduced the Cauchy Transform
to obtain improved $\ell_1$ embeddings, thereby leading to an algorithm for the
$\ell_1$ regression problem that runs in $\bigO(nd^{1.376+})$ time; and Clarkson
et al.~\cite{CDMMMW13_SODA} use the Fast Cauchy Transform and ellipsoidal rounding
methods to compute an approximation to the solution of general $\ell_p$
regression problems in roughly $\bigO(nd \log n)$ time.

These algorithms, and in particular the algorithms for $p=2$, form the basis for
much of the large body of recent work in randomized algorithms for low-rank
matrix approximation, and thus optimizing their properties can have immediate
practical benefits.
See, e.g., the recent monograph of Mahoney~\cite{Mah-mat-rev_BOOK} and
references therein for details.
Although some of these algorithms are near-optimal for dense inputs, they all
require $\Omega(nd \log d)$ time, which can be large if the input matrix is 
very sparse.
Thus, it was a significant result when Clarkson and
Woodruff~\cite{CW12sparse_TR} developed an algorithm for the $\ell_2$ regression
problem (as well as the related problems of low-rank matrix approximation and
$\ell_2$ leverage score approximation) that runs in \emph{input-sparsity time},
i.e., in $\bigO(\nnz(A) + \poly(d/\epsilon))$ time, where $\nnz(A)$ is the
number of non-zero elements in $A$ and $\epsilon$ is an error parameter.
This result depends on the construction of a \emph{sparse embedding matrix}
$\Pi$ for $\ell_2$.
By this,
we mean the following: 
for an $n \times d$ matrix $A$, an $s \times n$ matrix $\Pi$ such that, 
$$
(1-\epsilon)\|Ax\|_2 \le \|\Pi A x\|_2 \le (1+\epsilon)\|Ax\|_2 , 
$$
for all $x\in\R^{d}$.
That is, $\Pi$ embeds the column space of $A$ into $\R^{s}$, while approximately
preserving the $\ell_2$ norms of all vectors in that subspace.
Clarkson and Woodruff achieve their improved results for $\ell_2$-based problems
by showing how to construct such a $\Pi$ with $s=\poly(d/\epsilon)$ and showing
that it can be applied to an arbitrary $A$ in $\bigO(\nnz(A))$
time~\cite{CW12sparse_TR}.
(In particular, this embedding result improves the result of Meng, Saunders, and
Mahoney~\cite{MSM11_TR}, who in their development of the parallel least-squares
solver \textsc{LSRN} use a result from Davidson and
Szarek~\cite{davidson2001local} to construct a constant-distortion embedding for
$\ell_2$ that runs in $\bigO(\nnz(A) \cdot d)$ time.)
Interestingly, the analysis of Clarkson and Woodruff coupled ideas from the data
streaming literature with the structural fact that there cannot be too many
high-leverage constraints/rows in $A$.
In particular, they showed that the high-leverage parts of the subspace may be
viewed as heavy-hitters that are ``perfectly hashed,'' and thus contribute no
distortion, and that the distortion of the rest of the subspace as well as the
``cross terms'' may be bounded with a result of Dasgupta, Kumar, and
Sarl\'{o}s~\cite{DKT10}.

In this paper, we provide improved low-distortion subspace embeddings for
$\ell_p$, for all $p\in[1,2]$, in input-sparsity time; and we show that, by
coupling with recent work on fast subspace-preserving sampling
from~\cite{CDMMMW13_SODA}, these embeddings can be used to provide
$(1+\epsilon)$-approximate solutions to $\ell_p$ regression problems, for
$p\in[1,2]$, in nearly input-sparsity time.
In more detail, our main results are the following.
\begin{itemize}
\item For $\ell_2$, we obtain an improved result for the input-sparsity time
  $(1\pm\epsilon)$-distortion embedding of~\cite{CW12sparse_TR}.
  In particular, for the same embedding procedure, we obtain improved bounds 
  for the embedding dimension with a much simpler analysis 
  than~\cite{CW12sparse_TR}.
  See Theorem~\ref{thm:sparse_l2} of Section~\ref{sxn:l2} for a precise 
  statement of this result.
  Our analysis is direct and does \emph{not} rely on splitting the 
  high-dimensional space into a set of heavy-hitters consisting of the 
  high-leverage components and the complement of that heavy-hitting set.  
  In addition, since our result directly improves the $\ell_2$ embedding 
  result of Clarkson and Woodruff~\cite{CW12sparse_TR}, it immediately leads 
  to improvements for the $\ell_2$ regression, low-rank matrix approximation, 
  and $\ell_2$ leverage score estimation problems that they consider.
\item For $\ell_1$, we obtain a low-distortion sparse embedding matrix $\Pi$
  such that $\Pi A$ can be computed in input-sparsity time.
  That is, we construct an embedding matrix $\Pi \in
  \R^{\poly(d) \times n}$ such that, for all $x\in\R^{d}$,
  $$
  1/\bigO(\poly(d)) \cdot \|Ax\|_1 \le \|\Pi A x\|_1 \le \bigO(\poly(d)) \cdot
  \|Ax\|_1,
  $$
  with a constant probability, and $\Pi A$ can be computed in $\bigO(\nnz(A))$
  time.  
  See Theorem~\ref{thm:sparse_l1} of Section~\ref{sxn:l1} for a precise
  statement of this result.  
  Here, our proof involves splitting the set $Y = \{ U x \,|\,
  \|x\|_{\infty}=1,\ x\in\R^{d} \}$, where $U$ is an $\ell_1$ well-conditioned
  basis for the span of $A$, into two parts, informally a subset where
  coordinates of high $\ell_1$ leverage dominate $\|y\|_1$ and the complement of
  that subset. 
  This $\ell_1$ result leads to immediate improvements in $\ell_1$-based
  problems.
  For example, by taking advantage of the fast version of subspace-preserving
  sampling from~\cite{CDMMMW13_SODA}, we can construct and apply a
  $(1\pm\epsilon)$-distortion sparse embedding matrix for $\ell_1$ in
  $\bigO(\nnz(A) \cdot \log n + \poly(d/\epsilon))$~time. 
  In addition, we can use it to compute a $(1+\epsilon)$-approximation to the
  $\ell_1$ regression problem in $O(\nnz(A) \cdot \log n + \poly(d/\epsilon))$
  time, which in turn leads to immediate improvements in $\ell_1$-based matrix
  approximation objectives, e.g., for the $\ell_1$ subspace approximation
  problem~\cite{bd09,SW11,CDMMMW13_SODA}.
\item For $\ell_p$, for all $p\in (1,2)$, we obtain a low-distortion sparse
  embedding matrix $\Pi$ such that $\Pi A$ can be computed in input-sparsity
  time.
  That is, we construct an embedding matrix $\Pi \in \R^{\poly(d) \times n}$
  such that, for all $x\in\R^{d}$,
  $$
  1/\bigO(\poly(d)) \cdot \|Ax\|_p \le \|\Pi A x\|_p \le
  \bigO(\poly(d)) \cdot \|Ax\|_p ,
  $$
  with a constant probability, and $\Pi A$ can be computed in $\bigO(\nnz(A))$
  time.
  See Theorem~\ref{thm:sparse_lp} of Section~\ref{sxn:lp} for a precise
  statement of this result.
  Here, our proof generalizes the $\ell_1$ result, but we need to prove upper
  and lower tail bound inequalities for sampling from general $p$-stable
  distributions that are of independent interest.
  Although these distributions don't have closed forms for $p\in(1,2)$ in
  general, we prove that there exists an order among the Cauchy distribution, a
  $p$-stable distribution with $p\in(1,2)$, and the Gaussian distribution such
  that for all $p\in(1,2)$ we can use the upper bound from the Cauchy
  distribution and the lower bound from the Gaussian distribution.
  As with our $\ell_1$ result, this $\ell_p$ result has several extensions: in
  $\bigO(\nnz(A) \cdot \log n + \poly(d/\epsilon))$ time, we can construct and
  apply a $(1\pm\epsilon)$-distortion sparse embedding matrix for $\ell_p$; in
  $\bigO(\nnz(A) \cdot \log n + \poly(d/\epsilon))$ time, we can compute a
  $(1+\epsilon)$-approximation to the $\ell_p$ regression problem; and in
  $\bigO(\nnz(A) \cdot d \log d)$ time, we can construct and apply a
  near-optimal (in terms of embedding dimension and distortion factor) 
  embedding~matrix.
\end{itemize}

\noindent
The $(1\pm\epsilon)$-distortion subspace embedding (for $\ell_p$, $p\in[1,2)$,
that we construct from the input-sparsity time embedding and the fast
subspace-preserving sampling) has embedding dimension $s =
\bigO(\poly(d)\log(1/\epsilon)/\epsilon^2)$, where the somewhat large $\poly(d)$
term directly multiplies the $\log(1/\epsilon)/\epsilon^2$ term. 
We can also improve this, showing that it is possible, without increasing the
overall complexity, to decouple the large $\poly(d)$ and
$\log(1/\epsilon)/\epsilon^2$ via another round of sampling and conditioning,
thereby obtaining an embedding dimension that is a small $\poly(d)$ times
$\log(1/\epsilon)/\epsilon^2$. 
See Theorem~\ref{thm:improved-dim} of Section~\ref{sxn:improve} for a precise
statement of this result.  

\bigskip

\textbf{Remark.}
Subsequent to our posting a preliminary version of this paper on the
arXiv~\cite{MM12_TR}, Clarkson and Woodruff let us know that, independently of
us, they used a result from~\cite{CDMMMW13_SODA} to extend their $\ell_2$
subspace embedding from~\cite{CW12sparse_TR} to provide a nearly input-sparsity
time algorithm for $\ell_p$ regression, for all $p\in[1,\infty)$.
This is now posted as Version 2 of~\cite{CW12sparse_TR}.
Their approach requires solving a rounding problem of size $O(n/\poly(d)) \times
d$, which depends on $n$ (possibly very large).
Our approach does not contain this intermediate step and it only needs
$O(\poly(d))$ storage.
Moreover, to the best of our knowledge, their method does not provide
low-distortion $\ell_p$ subspace embeddings in input-sparsity time, as we are
able to provide (in a simple and oblivious way).

\textbf{Remark.}
In the first version of this paper, the embedding dimension for $\ell_2$ in
Theorem~\ref{thm:sparse_l2} was $\bigO(d^4/\epsilon^2)$.
Subsequent to the dissemination of this version, Drineas pointed out to us that,
with a slight modification to our original proof, our result could very easily
be improved to $\bigO(d^2/\epsilon^2)$.
Nelson and Nguyen also let us know that, at about the same time and using the
same technique, but independent of us, they too obtained and first published the
$\bigO(d^2/\epsilon^2)$ embedding result~\cite{nelson2012osnap}.

\section{Background}
\label{sxn:lbackground} 

We use $\|\cdot\|_p$ to denote the $\ell_p$ norm of a vector, $\|\cdot\|_2$ the
spectral norm of a matrix, $\|\cdot\|_F$ the Frobenius norm of a matrix, and
$|\cdot|_p$ the element-wise $\ell_p$ norm of a matrix.
Given $A \in \R^{n \times d}$ with full column rank and $p \in [1, 2]$, we 
use $\A_p$ to denote the $\ell_p$ subspace spanned by $A$'s columns.
In this paper, we are interested in fast embedding of $\A_p$ into a
$d$-dimensional subspace of $(\R^{\poly(d)}, \|\cdot\|_p)$, with distortion
either $\poly(d)$ or $(1\pm\epsilon)$, for some $\epsilon > 0$, as well as
applications of this embedding to problems such as $\ell_p$ regression.
We assume that $n \gg \poly(d) \geq d \gg \log n$.
To state our results, we assume that we are capable of computing a
$(1+\epsilon)$-approximate solution to an $\ell_p$ regression problem of size
$n' \times d$ for some $\epsilon > 0$, as long as $n'$ is independent of $n$.
Let us denote the running time needed to solve this smaller problem by
$\T_p(\epsilon; n', d)$.
In theory, we have $\T_2(\epsilon; n', d) = \bigO(n' d \log (d/\epsilon) + d^3)$
(see Rokhlin and Tygert~\cite{RT08} and Drineas et 
al.~\cite{DMMS07_FastL2_NM10}), and $\T_p(\epsilon; n', d) =
\bigO((n' d^2 + \poly(d)) \log(n'/\epsilon))$, for general $p$ (see, e.g.,
Mitchell~\cite{mitchell2003polynomial}).


\paragraph{Conditioning.}
The $\ell_p$ subspace embedding and $\ell_p$ regression problems are closely
related to the concept of conditioning.  
We state here two related notions of $\ell_p$-norm conditioning and then a lemma
that characterizes the relationship between them.

\begin{definition}[$\ell_p$-norm Conditioning (from \cite{CDMMMW13_SODA})]
  \label{def:lpnormcond}
  Given an $n \times d$ matrix $A$ and $p \in [1, \infty]$, let
  \begin{equation*}
    \sigma_p^{\max}(A) = \max_{\|x\|_2 \leq 1} \|A x\|_p \text{ and } \sigma_p^{\min}(A) = \min_{\|x\|_2 \geq 1} \|A x\|_p.
  \end{equation*}
  Then, we denote by $\kappa_p(A)$ the \emph{$\ell_p$-norm condition number of
    $A$}, defined to be:
  \begin{equation*}
    \kappa_p(A) = \sigma_p^{\max}(A) / \sigma_p^{\min}(A).
  \end{equation*}
  For simplicity, we will use $\kappa_p$, $\sigma_p^{\min}$, and 
  $\sigma_p^{\max}$ when the underlying matrix is clear.
\end{definition}

\begin{definition}[$(\alpha, \beta, p)$-conditioning (from \cite{DDHKM09_lp_SICOMP})]
  \label{def:lpbasis}
  Given an $n \times d$ matrix $A$ and $p\in[1,\infty]$, let $q$ be the dual
  norm of $p$.
  Then $A$ is \emph{$(\alpha,\beta,p)$-conditioned} if (1) $|A|_p \leq \alpha$,
  and (2) for all $z \in \R^{d}$, $\|z\|_q \leq \beta \|A z\|_p$.
  Define $\bar{\kappa}_p(A)$ as the minimum value of $\alpha \beta$ such that
  $A$ is $(\alpha, \beta, p)$-conditioned. 
\end{definition}

\begin{lemma}[Equivalence of $\kappa_p$ and $\bar{\kappa}_p$ (from
  \cite{CDMMMW13_SODA})]
  \label{lemma:kappa_equiv}
  Given an $n \times d$ matrix $A$ and $p \in [1, \infty]$, we always have
  \begin{equation*}
    d^{-|1/2-1/p|} \kappa_p(A) \leq \bar{\kappa}_p(A) \leq d^{\max \{1/2, 1/p\}} \kappa_p(A).
  \end{equation*}
\end{lemma}

\noindent
\textbf{Remark.}
Given the equivalence established by Lemma~\ref{lemma:kappa_equiv}, we will say
that $A$ is \emph{well-conditioned in the $\ell_p$ norm} if $\kappa_p(A)$ or
$\bar{\kappa}_p(A) = \bigO(\poly(d))$, independent of $n$. 

Although for an arbitrary matrix $A \in \R^{n \times d}$, the condition numbers
$\kappa_p(A)$ and $\bar{\kappa}_p(A)$ can be arbitrarily large, we can often
find a matrix $R \in \R^{d \times d}$ such that $A R^{-1}$ is well-conditioned. 
This procedure is called \emph{conditioning}, and there exist two approaches for
conditioning: via low-distortion $\ell_p$ subspace embedding and via ellipsoidal
rounding.

\begin{definition}[Low-distortion $\ell_p$ Subspace Embedding]
  Given an $n \times d$ matrix $A$ and $p \in [1, \infty]$, $\Pi \in \R^{s
    \times n}$ is a low-distortion embedding of $\A_p$ if $s = \bigO(\poly(d))$
  and
  \begin{equation*}
    1/\bigO(\poly(d)) \cdot \|A x\|_p \leq \|\Pi A x\|_p \leq \bigO(\poly(d)) \cdot \|A x\|_p, \quad \forall x \in \R^d.
  \end{equation*}
\end{definition}

\noindent
\textbf{Remark.}
Given a low-distortion embedding matrix $\Pi$ of $\A_p$, let $R$ be the ``R''
matrix from the QR decomposition of $\Pi A$. 
Then, the matrix $A R^{-1}$ is well-conditioned in the $\ell_p$ norm. 
To see this, note that we~have
\begin{align*}
  \|A R^{-1} x\|_p \leq \bigO(\poly(d)) \cdot \|\Pi A R^{-1} x\|_p 
  \leq \bigO(\poly(d)) \cdot \|\Pi A R^{-1} \|_2 
  = \bigO(\poly(d)) \cdot \|x\|_2, \quad \forall x \in \R^d,
\end{align*}
where the first inequality is due to low distortion and the second inequality is
due to $s = \bigO(\poly(d))$. 
By similar arguments, we can show that $\|A R^{-1} x\|_p \geq 1/\bigO(\poly(d))
\cdot \|x\|_2,\ \forall x \in\R^d$. 
Hence, by combining these results, the matrix $A R^{-1}$ is well-conditioned in
the $\ell_p$ norm. 

For a discussion of ellipsoidal rounding, we refer readers to Clarkson et
al.~\cite{CDMMMW13_SODA}. 
In this paper, we simply cite the following lemma, which is based on ellipsoidal
rounding.

\begin{lemma}[Fast $\bigO(d)$-conditioning (from \cite{CDMMMW13_SODA})]
  \label{lemma:lp_cond_2d}
  Given an $n \times d$ matrix $A$ and $p \in [1, \infty]$, it takes at most
  $\bigO(n d^3 \log n)$ time to find a matrix $R \in \R^{d \times d}$ such that
  $\kappa_p(A R^{-1}) \leq 2 d$.
\end{lemma}

\paragraph{Subspace-preserving sampling and $\ell_p$ regression.}
Given $R \in \R^{d \times d}$ such that $A R^{-1}$ is well-conditioned in the
$\ell_p$ norm, we can construct a $(1\pm\epsilon)$-distortion embedding,
specifically a subspace-preserving sampling, of $\A_p$ in $\bigO(\nnz(A) \cdot
\log n)$ additional time and with a constant probability.
This result from Clarkson et al.~\cite[Theorem 5.4]{CDMMMW13_SODA} improves 
the subspace-preserving sampling algorithm proposed by Dasgupta et 
al.~\cite{DDHKM09_lp_SICOMP} by estimating the row norms of $A R^{-1}$
(instead of computing them exactly) to define importance sampling 
probabilities.

\begin{lemma}[Fast Subspace-preserving Sampling (from \cite{CDMMMW13_SODA})]
  \label{lemma:fast_sampling}
  Given a matrix $A \in \R^{n \times d}$, $p \in [1, \infty)$, $\epsilon > 0$,
  and a matrix $R \in \R^{d \times d}$ such that $A R^{-1}$ is well-conditioned,
  it takes $\bigO(\nnz(A) \cdot \log n)$ time to compute a sampling matrix $S
  \in \R^{s \times n}$ (with only one nonzero element per row) with $s =
  \bigO(\bar{\kappa}_p^p(A R^{-1}) d^{|p/2-1|+1} \log(1/\epsilon) / \epsilon^2)$
  such that with a constant probability,
  \begin{equation*}
    (1-\epsilon) \|A x\|_p \leq \|S A x\|_p \leq (1+\epsilon) \|A x\|_p, \quad \forall x \in \R^d.
  \end{equation*}
\end{lemma}

\noindent
Given such a subspace-preserving sampling algorithm, Clarkson et
al.~\cite[Theorem 5.4]{CDMMMW13_SODA} show that it is straightforward to
compute a $\frac{1+\epsilon}{1-\epsilon}$-approximate solution to an $\ell_p$
regression problem.

\begin{lemma}[$\ell_p$ Regression via Sampling (from \cite{CDMMMW13_SODA}]
  \label{lemma:fast_reg}
  Given an $\ell_p$ regression problem specified by $A \in \R^{n \times d}$, $b
  \in \R^n$, and $p \in [1, \infty)$, let $S$ be a $(1\pm\epsilon)$-distortion
  embedding matrix of the subspace spanned by $A$'s columns and $b$ from
  Lemma~\ref{lemma:fast_sampling}, and let $\hat{x}$ be an optimal solution to
  the subsampled problem $\min_{x \in \R^d} \|S A x - S b\|_p$.
  Then $\hat{x}$ is a $\frac{1+\epsilon}{1-\epsilon}$-approximate solution to
  the original problem.
\end{lemma}

\noindent
\textbf{Remark.}
Collecting these results, we see that a low-distortion $\ell_p$ subspace
embedding is a fundamental building block (and very likely a bottleneck) for
$(1\pm\epsilon)$-distortion $\ell_p$ subspace embeddings, as well as for a
$(1+\epsilon)$-approximation to an $\ell_p$ regression problem. 
This motivates our work and its emphasis on finding low-distortion subspace 
embeddings more efficiently.

\paragraph{Stable distributions.}
\label{sec:stable-distributions}

The properties of $p$-stable distributions are essential for constructing
input-sparsity time low-distortion $\ell_p$ subspace embeddings.

\begin{definition}[$p$-stable Distribution]
  A distribution $\D$ over $\R$ is called $p$-stable, if for any $m$ real
  numbers $a_1,\ldots,a_m$, we have
  \begin{equation*}
    \sum_{i=1}^m a_i X_i \simeq \left( \sum_{i=1}^m |a_i|^p \right)^{1/p} X,
  \end{equation*}
  where $X_i \stackrel{\text{iid}}{\sim} \D$ and $X \sim \D$.
  By ``$X \simeq Y$'', we mean $X$ and $Y$ have the same distribution.
\end{definition}

\noindent
By a result due to L{\'e}vy~\cite{levy1925calcul}, it is known that $p$-stable
distributions exist for $p \in (0, 2]$; and from Chambers et
al.~\cite{chambers1976method}, it is known that $p$-stable random variables can
be generated efficiently, thus allowing their practical use.
Let us use $\D_p$ to denote the ``standard'' $p$-stable distribution, for
$p\in[1,2]$, specified by its characteristic function $\psi(t) = e^{-|t|^p}$.
It is known that $\D_1$ is the standard Cauchy distribution, and that $\D_2$ is
the Gaussian distribution with mean $0$ and variance $2$.

\paragraph{Tail inequalities.}
We note two inequalities from Clarkson et al.~\cite{CDMMMW13_SODA} regarding the
tails of the Cauchy distribution.

\begin{lemma}[Cauchy Upper Tail Inequality]
  \label{lemma:cauchy_upper}
  For $i=1,\ldots,m$, let $C_i$ be $m$ (not necessarily independent) standard
  Cauchy variables, and $\gamma_i > 0$ with $\gamma = \sum_i \gamma_i$.
  Let $X = \sum_i \gamma_i |C_i|$.
  For any $t > 1$,
  \begin{equation*}
    \Pr[X > t \gamma] \leq \frac{1}{\pi t} \left( \frac{\log(1+(2mt)^2)}{1 - 1/(\pi t)} + 1 \right).
  \end{equation*}
  For simplicity, we assume that $m \geq 3$ and $t \geq 1$, and then we have
  $\Pr[X > t \gamma] \leq 2 \log (m t)/ t$.
\end{lemma}

\begin{lemma}[Cauchy Lower Tail Inequality]
  \label{lemma:cauchy_lower}
  For $i=1,\ldots,m$, let $C_i$ be independent standard Cauchy random variables, and
  $\gamma_i \geq 0$ with $\gamma = \sum_i \gamma_i$. 
  Let $X = \sum_i \gamma_i |C_i|$. 
  Then, for any $t>0$,
  \begin{equation*}
    \log \Pr[X \leq (1-t) \gamma] \leq \frac{- \gamma t^2}{3 \max_i \gamma_i}.
  \end{equation*}
\end{lemma}

\noindent
We also note the following result about Gaussian variables.
This is a direct consequence of Maurer's inequality~(\cite{maurer2003bound}),
and we will use it to derive lower tail inequalities for $p$-stable
distributions.

\begin{lemma}[Gaussian Lower Tail Inequality]
  \label{lemma:gaussian_lower}
  For $i=1,\ldots,m$, let $G_i$ be independent standard Gaussian random
  variables, and $\gamma_i \geq 0$ with $\gamma = \sum_i \gamma_i$. 
  Let $X = \sum_i \gamma_i |G_i|^2$. 
  Then, for any $t>0$,
  \begin{equation*}
    \log \Pr[X \leq (1-t) \gamma] \leq \frac{-\gamma t^2}{6 \max_i \gamma_i}.
  \end{equation*}
\end{lemma}


\section{Main Results for $\ell_2$ Embedding}
\label{sxn:l2} 

Here is our main result for input-sparsity time low-distortion subspace
embeddings for $\ell_2$.
See also Nelson and Nguyen~\cite{nelson2012osnap} for a similar result with a
slightly better constant.

\begin{theorem}[($1\pm\epsilon$)-distortion Embedding for $\ell_2$]
  \label{thm:sparse_l2}
  Given a matrix $A \in \R^{n \times d}$ and $\epsilon \in (0, 1)$, let $\Pi = S
  D$ where $S \in \R^{s \times n}$ has each column chosen independently and
  uniformly from the $s$ standard basis vectors of $\R^s$ and $D \in \R^{n
    \times n}$ is a diagonal matrix with diagonal entries chosen independently
  and uniformly from $\pm 1$.
  Given any $\delta \in (0, 1)$, let $s = (d^2 + d) / (\epsilon^2 \delta)$.
  Then with probability at least $1-\delta$,
  \begin{equation*}
    (1-\epsilon) \|A x\|_2 \leq \| \Pi A x \|_2 \leq (1+\epsilon) \|A x\|_2, \quad \forall x \in \R^d.
  \end{equation*}
  In addition, $\Pi A$ can be computed in $\bigO(\nnz(A))$ time.
\end{theorem}

\noindent
The construction of $\Pi$ in this theorem is the same as the construction in
Clarkson and Woodruff~\cite{CW12sparse_TR}.
For them, $s = \bigO((d/\epsilon)^4 \log^2(d/\epsilon))$ in order to achieve $(1
\pm \epsilon)$ distortion with a constant probability. 
Theorem~\ref{thm:sparse_l2} shows that it actually suffices to set $s =
\bigO((d^2 + d)/\epsilon^2)$.
Surprisingly, the proof is rather simple.
Let $X = U^T \Pi^T \Pi U$, where $U$ is an orthonormal basis for $\A_2$.
Compute $\mathbf{E}[\|X-I\|_F^2]$ and apply Markov's inequality to $\|X -
I\|_F^2 \leq \epsilon^2$, which implies $\|X - I\|_2 \leq \epsilon$ and hence
the embedding result.
See Appendix~\ref{sec:proof_l2} for a complete proof. 

\noindent
\textbf{Remark.} 
The $\bigO(\nnz(A))$ running time is indeed optimal, up to constant factors, 
for general inputs.
Consider the case when $A$ has an important row $a_j$ such that $A$ becomes
rank-deficient without it.
Thus, we have to observe $a_j$ in order to compute a low-distortion embedding.
However, without any prior knowledge, we have to scan at least a constant
portion of the input to guarantee that $a_j$ is observed with a constant
probability, which takes $\bigO(\nnz(A))$ time.
Note that this optimality result applies to general $p$.

The results of Theorem~\ref{thm:sparse_l2} propagate to related applications,
e.g., to the $\ell_2$ regression problem, the low-rank matrix approximation
problem and the problem of computing approximations to the $\ell_2$ leverage
scores.
Since it underlies the other applications, only the $\ell_2$ regression
improvement is stated here explicitly; its proof is basically combining our
Theorem~\ref{thm:sparse_l2} with Theorem~19 of~\cite{CW12sparse_TR}.

\begin{corollary}[Fast $\ell_2$ Regression]
  With a constant probability, a $(1+\epsilon)$-approximate solution to an
  $\ell_2$ regression problem can be computed in $\bigO(\nnz(A) + \T_2(\epsilon;
  d^2 / \epsilon^2, d))$ time.
\end{corollary}

\noindent
\textbf{Remark.}  
Although our simpler direct proof leads to a better result for $\ell_2$ subspace
embedding, the technique used in the proof of Clarkson and
Woodruff~\cite{CW12sparse_TR}, which splits coordinates into ``heavy'' and
``light'' sets based on the leverage scores, highlights an important structural
property of $\ell_2$ subspace: that only a small subset of coordinates can have
large $\ell_2$ leverage scores.  
(We note that the technique of splitting coordinates is also used by Ailon and
Liberty~\cite{AL11} to get an unrestricted fast Johnson-Lindenstrauss transform;
and that the difficulty in finding and approximating the large-leverage
directions was---until
recently~\cite{Mah-mat-rev_BOOK,DMMW12_ICML}---responsible for difficulties in
obtaining fast relative-error random sampling algorithms for $\ell_2$ regression
and low-rank matrix approximation.)  
An analogous structural fact holds for $\ell_1$ and other $\ell_p$ spaces.  
Using this property, we can construct novel input-sparsity time $\ell_p$ 
subspace embeddings for general $p \in [1, 2)$, as we discuss in the next 
two~sections.

\section{Main Results for $\ell_1$ Embedding}
\label{sxn:l1} 

Here is our main result for input-sparsity time low-distortion subspace
embeddings for $\ell_1$.

\begin{theorem}[Low-distortion Embedding for $\ell_1$]
  \label{thm:sparse_l1}
  Given $A \in \R^{n \times d}$ with full column rank, let $\Pi = S C \in \R^{s
    \times n}$, where $S \in \R^{s \times n}$ has each column chosen
  independently and uniformly from the $s$ standard basis vectors of $\R^s$, and
  where $C \in \R^{n \times n}$ is a diagonal matrix with diagonals chosen
  independently from the standard Cauchy distribution.
  Set $s = \omega d^5 \log^5 d$ with $\omega$ sufficiently large.
  Then with a constant probability, we have
  \begin{equation*}
    1/\bigO(d^{2} \log^2 d) \cdot \|A x\|_1 \leq \|\Pi A x\|_1 \leq \bigO( d \log d ) \cdot \|A x\|_1, \quad \forall x \in \R^{d}.
  \end{equation*}
  In addition, $\Pi A$ can be computed in $\bigO(\nnz(A))$ time.
\end{theorem}

\noindent
The construction of the $\ell_1$ subspace embedding matrix is different than its
$\ell_2$ norm counterpart only by the diagonal elements of $D$ (or $C$): whereas
we use $\pm 1$ for the $\ell_2$ norm, we use Cauchy variables for the $\ell_1$
norm. 
The proof of Theorem~\ref{thm:sparse_l1} uses the technique of splitting
coordinates, the fact that the Cauchy distribution is $1$-stable, and the upper
and lower tail tail inequalities regarding the Cauchy distribution from
Lemmas~\ref{lemma:cauchy_upper} and~\ref{lemma:cauchy_lower}.
See Appendix~\ref{sec:proof_l1} for a complete proof.

\noindent
\textbf{Remark.} 
As mentioned above, the $\bigO(\nnz(A))$ running time is optimal.
Whether the distortion $\bigO(d^3 \log^3 d)$ is optimal is still an open
question. 
However, for the same construction of $\Pi$, we can provide a ``bad'' case that
provides a lower bound.
Choose $A =
\begin{pmatrix}
  I_d & \mathbf{0}
\end{pmatrix}^T$.
Suppose that $s$ is sufficiently large such that with an overwhelming
probability, the top $d$ rows of $A$ are perfectly hashed, i.e., $\|\Pi A x\|_1
= \sum_{k=1}^d |c_k||x_k|$, $\forall x \in \R^d$, where $c_k$ is the $k$-th
diagonal of $C$.
Then, the distortion of $\Pi$ is $\max_{k \leq d} |c_k| / \min_{k \leq d} |c_k|
\approx \bigO(d^2)$. 
Therefore, at most an $\bigO(d \log^3 d)$ factor of the distortion is due to
artifacts in our analysis.

Our input-sparsity time $\ell_1$ subspace embedding of
Theorem~\ref{thm:sparse_l1} improves the $\bigO(\nnz(A) \cdot d \log d)$-time
embedding by Sohler and Woodruff~\cite{SW11} and the $\bigO(n d \log n)$-time 
embedding of Clarkson et al.~\cite{CDMMMW13_SODA}.
In addition, by combining Theorem~\ref{thm:sparse_l1} and
Lemma~\ref{lemma:fast_sampling}, we can compute a $(1\pm\epsilon)$-distortion
embedding in $\bigO(\nnz(A) \cdot \log n)$ time, i.e., in \emph{nearly} 
input-sparsity time.

\begin{theorem}[($1\pm\epsilon$)-distortion Embedding for $\ell_1$]
  \label{thm:sparse_l1-eps}
  Given $A \in \R^{n \times d}$, it takes $\bigO(\nnz(A) \cdot \log n)$ time to
  compute a sampling matrix $S \in \R^{s \times n}$ with $s = \bigO(\poly(d)
  \log(1/\epsilon) / \epsilon^2)$ such that with a constant probability, $S$
  embeds $\A_1$ into $(\R^{s}, \|\cdot\|_1)$ with distortion $1 \pm \epsilon$.
\end{theorem}

Our improvements in Theorems~\ref{thm:sparse_l1} and~\ref{thm:sparse_l1-eps}
also propagate to related $\ell_1$-based applications, including the $\ell_1$
regression and the $\ell_1$ subspace approximation problem considered
in~\cite{SW11,CDMMMW13_SODA}.
As before, only the regression improvement is stated here explicitly.
For completeness, we present in Algorithm \ref{alg:fast_l1_reg} our algorithm
for solving $\ell_1$ regression problems in nearly input-sparsity time.
The brief proof of Corollary~\ref{cor:l1reg}, our main quality-of-approximation
result for Algorithm~\ref{alg:fast_l1_reg}, may be found in
Appendix~\ref{sxn:pf-cor-l1reg}.

\begin{algorithm}
  \caption{Fast $\ell_1$ Regression Approximation in $\bigO(\nnz(A) \cdot \log n
    + \poly(d) \log(1/\epsilon)/\epsilon^2)$ Time}
  \label{alg:fast_l1_reg}
  \begin{algorithmic}[1]
    \Require $A \in \R^{n \times d}$ with full column rank, $b \in \R^n$, and
    $\epsilon \in (0, 1/2)$.  

    \Ensure A $(1+\epsilon)$-approximation solution $\hat{x}$ to $\min_{x \in
      \R^d} \|A x - b\|_1$, with a constant probability.

    \State Let $\bar{A} =
    \begin{pmatrix}
      A & b
    \end{pmatrix}$ and denote $\bar{\A}_1$ the $\ell_1$ subspace spanned by
    $A$'s columns and $b$.
    
    \State Compute a low-distortion embedding $\Pi \in \R^{\bigO(\poly(d))
      \times n}$ of $\bar{\A}_1$ (Theorem~\ref{thm:sparse_l1}).
    
    \State Compute $\bar{R} \in \R^{(d+1) \times (d+1)}$ from $\Pi \bar{A}$ such
    that $\bar{A} \bar{R}^{-1}$ is well-conditioned (QR or
    Lemma~\ref{lemma:lp_cond_2d}).

    \State Compute a $(1 \pm \epsilon/4)$-distortion embedding $S \in
    \R^{\bigO(\poly(d) \log(1/\epsilon)/\epsilon^2) \times n}$ of $\bar{\A}_1$
    (Lemma~\ref{lemma:fast_sampling}).

    \State Compute a $(1+\epsilon/4)$-approximate solution $\hat{x}$ to $\min_{x
      \in \R^d} \|S A x - S b\|_1$.
  \end{algorithmic}
\end{algorithm}

\begin{corollary}[Fast $\ell_1$ Regression]
  \label{cor:l1reg}
  With a constant probability, Algorithm~\ref{alg:fast_l1_reg} computes a
  $(1+\epsilon)$-approximate solution to an $\ell_1$ regression problem in
  $\bigO(\nnz(A) \cdot \log n + \T_1(\epsilon; \poly(d) \log(1/\epsilon) /
  \epsilon^2, d))$ time.
\end{corollary}

\noindent
\textbf{Remark.}
For readers familiar with the impossibility results for dimension reduction in
$\ell_1$~\cite{CS02,LN04,BC05}, note that those results apply to arbitrary point
sets of size $n$ and are interested in embeddings that are ``oblivious,'' in
that they do not depend on the input data.
In this paper, we only consider points in a subspace, and the
subspace-preserving sampling procedure of~\cite{DDHKM09_lp_SICOMP} that we use
is data-dependent.  

\section{Main Results for $\ell_p$ Embedding}
\label{sxn:lp} 

In this section, we use the properties of $p$-stable distributions to generalize
the input-sparsity time $\ell_1$ subspace embedding to $\ell_p$ norms, for $p
\in (1, 2)$. 
Generally, $\D_p$ does not have explicit PDF/CDF, which increases the difficulty
for theoretical analysis. 
Indeed, the main technical difficulty here is that we are not aware of $\ell_p$
analogues of Lemmas~\ref{lemma:cauchy_upper} and~\ref{lemma:cauchy_lower} that
would provide upper and lower tail inequality for $p$-stable distributions.
(Indeed, even Lemmas~\ref{lemma:cauchy_upper} and~\ref{lemma:cauchy_lower} were
established only recently~\cite{CDMMMW13_SODA}.)

Instead of analyzing $\D_p$ directly, for any $p \in (1, 2)$, we
establish an order among the Cauchy distribution, the $p$-stable distribution,
and the Gaussian distribution, and then we derive upper and lower tail
inequalities for the $p$-stable distribution similar to the ones we used to
prove Theorem~\ref{thm:sparse_l1}.
We state these technical results here since they are of independent interest.
We start with the following lemma, which is proved in
Appendix~\ref{sxn:pf-equiv} and which establishes this order.

\begin{lemma}
  \label{lemma:equiv}
  For any $p \in (1, 2)$, there exist constants $\alpha_p > 0$ and $\beta_p > 0$
  such that
  \begin{equation*}
    \alpha_p |C| \succeq |X_p|^p \succeq \beta_p |G|^2, \quad 
  \end{equation*}
  where $C$ is a standard Cauchy variable, $X_p \sim \D_p$, $G$ is a standard
  Gaussian variable.
  By ``$X \succeq Y$'' we mean $\Pr[X \geq t] \geq \Pr[Y \geq t],\ \forall t\in
  \R$, i.e., $F_X(t) \leq F_Y(t),\ \forall t \in \R$, where $F(\cdot)$ is the
  corresponding~CDF.
\end{lemma}

\noindent
Our numerical results suggest that the constants $\alpha_p$ and $\beta_p$ 
are not too far away from $1$.
See Figure~\ref{fig:stable_cdf}, which plots of the CDFs of $|X_p/2|^p$ for 
$p=1,0, 1.1, \ldots, 2.0$, based on which we conjecture 
$|X_{p_1}/2|^{p_1} \succeq |X_{p_2}/2|^{p_2}$, for all 
$1 \leq p_1 \leq p_2 \leq 2$.
This implies that $2^{p-1} |C| \succeq |X_p|^p$ and 
$|X_p|^p \succeq 2^{p-2} |X_2|^2 \simeq 2^{p-1} |G|^2$, 
which therefore provides a value for the constants $\alpha_p$ and $\beta_p$.

\begin{figure}
  \centering
  \includegraphics[width=0.7\textwidth]{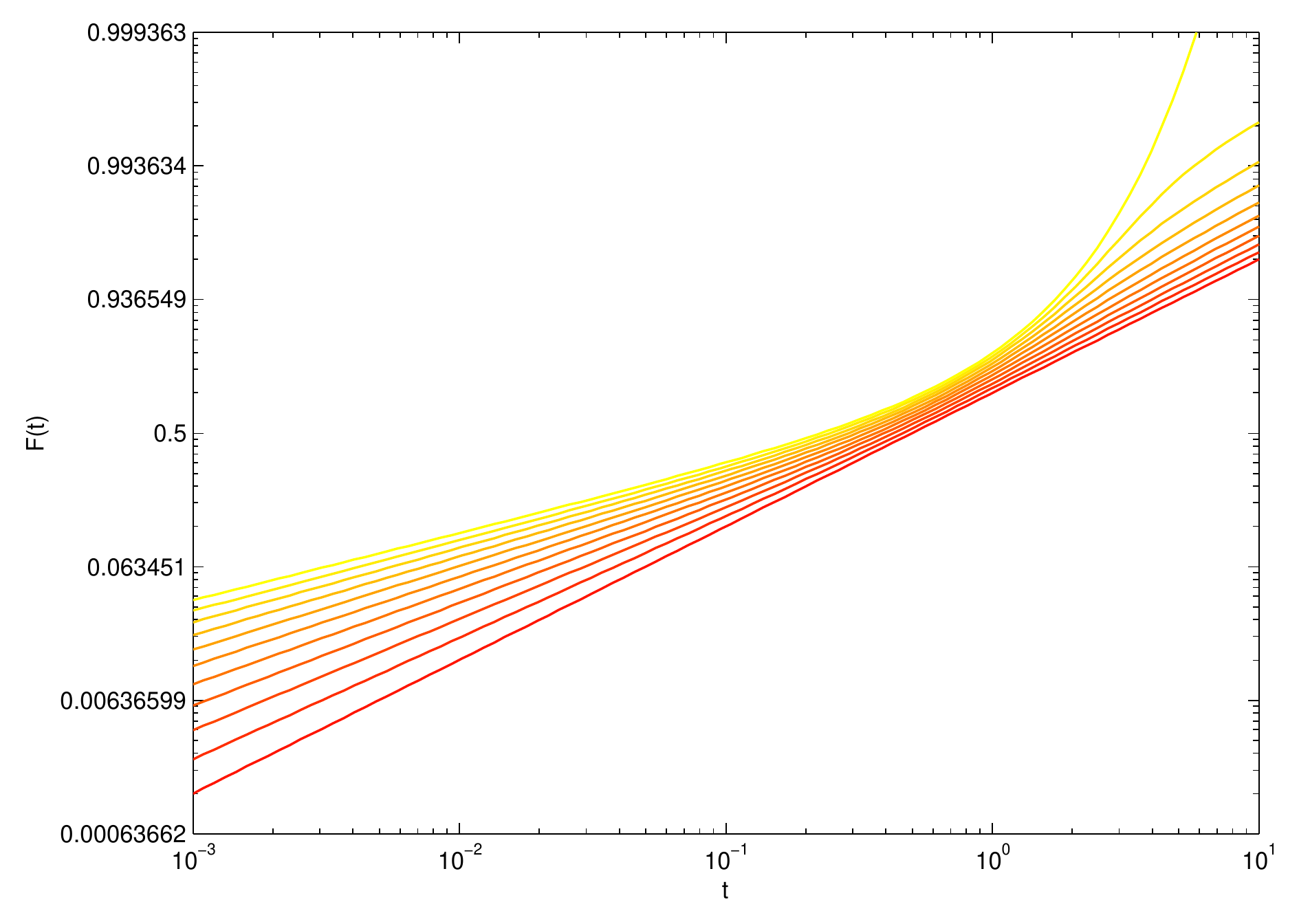}
  \caption{The CDFs ($F(t)$) of $|X_p/2|^p$ for $p = 1.0$ (bottom, i.e., red or
    dark gray), $1.1, \ldots, 2.0$ (top, i.e., yellow or light gray), where $X_p
    \sim \D_p$ and the scales of the axes are chosen to magnify the upper (as $t
    \to \infty$) and lower (as $t \to 0$) tails.
    These empirical results suggest $|X_{p_1}/2|^{p_1} \succeq
    |X_{p_2}/2|^{p_2}$ for all $1 \leq p_1 \leq p_2 \leq 2$.}
  \label{fig:stable_cdf}
\end{figure}

Lemma~\ref{lemma:equiv} suggests that we can use Lemma~\ref{lemma:cauchy_upper}
(regarding Cauchy random variables) to derive upper tail inequalities for
general $p$-stable distributions and that we can use
Lemma~\ref{lemma:gaussian_lower} (regarding Gaussian variables) to derive lower
tail inequalities for general $p$-stable distributions.
The following two lemmas establish these results; the proofs of these lemmas are
provided in Appendix~\ref{sxn:pf-lp-stable-upper} and
Appendix~\ref{sxn:pf-lp-stable-lower}, respectively.

\begin{lemma}[Upper Tail Inequality for $p$-stable Distributions]
  \label{lemma:stable_upper}
  Given $p \in (1, 2)$, for $i=1,\ldots,m$, let $X_i$ be $m$ (not necessarily
  independent) random variables sampled from $\D_p$, and $\gamma_i > 0$ with
  $\gamma = \sum_i \gamma_i$.
  Let $X = \sum_i \gamma_i |X_i|^p$.
  Assume that $m \geq 3$.
  Then for any $t \geq 1$,
  \begin{equation*}
    \Pr[X \geq t \alpha_p \gamma] \leq \frac{2 \log (mt)}{t}.
  \end{equation*}
\end{lemma}

\begin{lemma}[Lower Tail Inequality for $p$-stable Distributions]
  \label{lemma:stable_lower}
  For $i=1,\ldots,m$, let $X_i$ be independent random variables sampled from
  $\D_p$, and $\gamma_i \geq 0$ with $\gamma = \sum_i \gamma_i$. Let $X
  = \sum_i \gamma_i |c_i|$. Then,
  \begin{equation*}
    \log \Pr[X \leq (1-t) \beta_p \gamma] \leq \frac{- \gamma t^2}{6 \max_i \gamma_i}.
  \end{equation*}
\end{lemma}

Given these results, here is our main result for input-sparsity time 
low-distortion subspace embeddings for $\ell_p$.
The proof of this theorem is similar to the proof of 
Theorem~\ref{thm:sparse_l1}, except that we replace the $\ell_1$ norm 
$\|\cdot\|_1$ by $\|\cdot\|_p^p$ and use the tail inequalities from 
Lemmas~\ref{lemma:stable_upper} and~\ref{lemma:stable_lower} (rather than 
Lemmas~\ref{lemma:cauchy_upper} and~\ref{lemma:cauchy_lower}).  

\begin{theorem}[Low-distortion Embedding for $\ell_p$]
  \label{thm:sparse_lp}
  Given $A \in \R^{n \times d}$ with full column rank and $p\in(1,2)$, let $\Pi
  = S D \in \R^{s \times n}$ where $S \in \R^{s \times n}$ has each column
  chosen independently and uniformly from the $s$ standard basis vectors of
  $\R^s$, and where $D \in \R^{n \times n}$ is a diagonal matrix with diagonals
  chosen independently from $\D_p$.
  Set $s = \omega d^5 \log^5 d$ with $\omega$ sufficiently large.
  Then with a constant probability, we have
  \begin{equation*}
    1/\bigO((d \log d)^{2/p}) \cdot \|A x\|_p \leq \|\Pi A x\|_p \leq 
    \bigO( (d \log d)^{1/p} ) \cdot \|A x\|_p, \quad \forall x \in \R^{d}.
  \end{equation*}
  In addition, $\Pi A$ can be computed in $\bigO(\nnz(A))$ time.
\end{theorem}

\noindent
Similar to the $\ell_1$ case, our input-sparsity time $\ell_p$ subspace
embedding of Theorem~\ref{thm:sparse_lp} improves the $\bigO(n d \log n)$-time
embedding of Clarkson et al.~\cite{CDMMMW13_SODA}.
As we mentioned in Section~\ref{sxn:intro}, their construction (and hence the
construction of \cite{CW12sparse_TR}) works for all $p \in [1, \infty)$, but it
requires solving a rounding problem of size $\bigO(n/\poly(d)) \times d$ as an
intermediate step, which may become intractable when $n$ is very large in a
streaming environment, while our construction only needs $\bigO(\poly(d))$
storage.
By combining Theorem~\ref{thm:sparse_lp} and
Lemma~\ref{lemma:fast_sampling}, we can compute a $(1\pm\epsilon)$-distortion
embedding in $\bigO(\nnz(A) \cdot \log n)$ time.

\begin{theorem}[($1\pm\epsilon$)-distortion Embedding for $\ell_p$]
  \label{thm:sparse_lp-eps}
  Given $A \in \R^{n \times d}$ and $p \in [1, 2)$, it takes $\bigO(\nnz(A)
  \cdot \log n)$ time to compute a sampling matrix $S \in \R^{s \times n}$ with
  $s = \bigO(\poly(d) \log(1/\epsilon) / \epsilon^2)$ such that with a constant
  probability, $S$ embeds $\A_p$ into $(\R^{s}, \|\cdot\|_p)$ with distortion $1
  \pm \epsilon$.
\end{theorem}

\noindent
These improvements for $\ell_p$ subspace embedding also propagate to related
$\ell_p$-based applications.
In particular, we can establish an improved algorithm for solving the $\ell_p$
regression problem in nearly input-sparsity time.


\begin{corollary}[Fast $\ell_p$ Regression]
  Given $p \in (1, 2)$, with a constant probability, a
  $(1+\epsilon)$-approximate solution to an $\ell_p$ regression problem can be
  computed in
  $$
  \bigO(\nnz(A) \cdot \log n + \T_p(\epsilon; \poly(d) \log(1/\epsilon) / \epsilon^2, d))
  $$
  time.
\end{corollary}



For completeness, we also present a result for low-distortion dense 
embeddings for $\ell_p$ that the tail inequalities from 
Lemmas~\ref{lemma:stable_upper} and~\ref{lemma:stable_lower} enable us to 
construct.
See Appendix~\ref{sxn:pf-dense_lp} for a proof of the following theorem.

\begin{theorem}[Low-distortion Dense Embedding for $\ell_p$]
  \label{thm:dense_lp}
  Given $A \in \R^{n \times d}$ with full column rank and $p \in (1, 2)$, let
  $\Pi \in \R^{s \times n}$ whose entries are i.i.d.\ samples from
  $\D_p$.
  If $s = \omega d \log d$ for $\omega$ sufficiently large, with a constant
  probability, we have
  \begin{equation*}
    1/\bigO(1) \cdot \|A x\|_p \leq \|\Pi A x\|_p \leq \bigO((d \log d)^{1/p}) \cdot \|A x\|_p, \quad \forall x \in \R^d.
  \end{equation*}
  In addition, $\Pi A$ can be computed in 
  $\bigO(\nnz(A) \cdot d \log d)$ time.
\end{theorem}

\noindent
\textbf{Remark.}
The result in Theorem~\ref{thm:dense_lp} is based on a dense $\ell_p$ subspace
embeddings that is analogous to the dense Gaussian embedding for $\ell_2$ and
the dense Cauchy embedding of~\cite{SW11} for $\ell_1$.
Although the running time (if one is simply interested in FLOP counts in RAM) of
Theorem~\ref{thm:dense_lp} is somewhat worse than that of
Theorem~\ref{thm:sparse_lp}, the embedding dimension and condition number
quality (the ratio of the upper bound on the distortion and the lower bound on
the distortion) are much better.
Our numerical implementations, both with the $\ell_1$ norm~\cite{CDMMMW13_SODA}
and with the $\ell_2$ norm~\cite{MSM11_TR}, strongly suggest that the latter
quantities are more important to control when implementing randomized regression
algorithms in large-scale parallel and distributed settings.

\section{Improving the Embedding Dimension}
\label{sxn:improve} 

In Theorem~\ref{thm:sparse_l1} and Theorem~\ref{thm:sparse_lp}, the embedding
dimension is $s = \bigO(\poly(d) \log(1/\epsilon)/\epsilon^2)$, where the
$\poly(d)$ term is a somewhat large polynomial of $d$ that directly multiplies
the $\log(1/\epsilon)/\epsilon^2$ term.  
(See the remark below for comments on the precise value of the $\poly(d)$ term.)
This is not ideal for the subspace embedding and the $\ell_p$ regression,
because we want to have a small embedding dimension and a small subsampled
problem, respectively.
Here, we show that it is possible to decouple the large polynomial of $d$ and
the $\log(1/\epsilon)/\epsilon^2$ term via another round of sampling and
conditioning without increasing the complexity.
See Algorithm~\ref{alg:dim} for details on this procedure.
Theorem~\ref{thm:improved-dim} provides our main quality-of-approximation result
for Algorithm~\ref{alg:dim}; its proof can be found in
Appendix~\ref{sxn:pf-imp-dim}.

\begin{algorithm}
  \caption{Improving the Embedding Dimension}
  \label{alg:dim}
  \begin{algorithmic}[1]
    \Require $A \in \R^{n \times d}$ with full column rank, $p \in [1, 2)$, and
    $\epsilon \in (0, 1)$.

    \Ensure A $(1\pm\epsilon)$-distortion embedding $S \in
    \R^{\bigO(d^{3 + p/2} \log(1/\epsilon)/\epsilon^2) \times n}$ of $\A_p$.

    \State Compute a low-distortion embedding $\tilde{\Pi} \in
    \R^{\bigO(\poly(d)) \times n}$ of $\A_p$ (Theorems~\ref{thm:sparse_l1} and
    \ref{thm:sparse_lp}).
    
    \State Compute $\tilde{R} \in \R^{d \times d}$ from $\tilde{\Pi} A$ such
    that $A \tilde{R}^{-1}$ is well-conditioned (QR or
    Lemma~\ref{lemma:lp_cond_2d}).

    \State Compute a $(1 \pm 1/2)$-distortion embedding $\tilde{S} \in
    \R^{\bigO(\poly(d) \times n)}$ of $\A_p$ (Lemma~\ref{lemma:fast_sampling}).
    
    \State Compute $R \in \R^{d \times d}$ such that $\kappa_p(\tilde{S} A
    R^{-1}) \leq 2 d$ (Theorem~\ref{lemma:lp_cond_2d}).

    \State Compute a $(1 \pm \epsilon)$-distortion embedding $S \in
    \R^{\bigO(d^{3+p/2} \log(1/\epsilon) / \epsilon^2 )\times n}$ of $\A_p$
    (Lemma~\ref{lemma:fast_sampling}).
  \end{algorithmic}
\end{algorithm}

\begin{theorem}[Improving the Embedding Dimension]
  \label{thm:improved-dim}
  Given $p \in [1, 2)$, with a constant probability, Algorithm~\ref{alg:dim}
  computes a $(1\pm\epsilon)$-distortion embedding of $\A_p$ into
  $(\R^{\bigO(d^{3+p/2}\log(1/\epsilon)/\epsilon^2)}, \|\cdot\|_p$) in
  $\bigO(\nnz(A) \cdot \log n)$ time.
\end{theorem}

\noindent
Then, by applying Theorem~\ref{thm:improved-dim} to the $\ell_p$ regression
problem, we can improve the size of the subsampled problem and hence the overall
running time.

\begin{corollary}[Improved Fast $\ell_p$ Regression]
  Given $p \in [1, 2)$, with a constant probability, a
  $(1+\epsilon)$-approximate solution to an $\ell_p$ regression problem can be
  computed in
  \begin{equation*}
    \bigO(\nnz(A) \cdot \log n + \T_p(\epsilon; d^{3+p/2} \log(1/\epsilon) / \epsilon^2, d))
  \end{equation*}
  time.
  The second term comes from solving a subsampled problem of size
  $\bigO(d^{3+p/2} \log(1/\epsilon) / \epsilon^2) \times d$.
\end{corollary}

\noindent
\textbf{Remark.}
We have stated our results in the previous sections as $\poly(d)$ without
stating the value of the polynomial because there are numerous trade-offs
between the conditioning quality and the running time. 
For example, let $p = 1$.
We can use a rounding algorithm instead of QR to compute the $R$ matrix. 
If we use the input-sparsity time embedding with the $\bigO(d)$-rounding
algorithm of~\cite{CDMMMW13_SODA}, then the running time to compute the
$(1\pm\epsilon)$-distortion embedding is $\bigO(\nnz(A) \cdot \log n + d^8 /
\epsilon^2)$ and the embedding dimension is $\bigO(d^{6.5}/\epsilon^2)$
(ignoring $\log$ factors). 
If, on the other hand, we use QR to compute $R$, then the running time is
$\bigO(\nnz(A) \cdot \log n + d^7/\epsilon^2)$ and the embedding dimension is
$\bigO(d^8/\epsilon^2)$. 
However, with the result from this section, the running time is simply
$\bigO(\nnz(A) \cdot \log n + \poly(d) + \T_p(\epsilon; d^{3+p/2}/\epsilon^2,
d))$ and the $\poly(d)$ term can be absorbed by the $\nnz(A)$ term.

\section{Acknowledgments}

The authors want to thank Petros~Drineas for reading a preliminary version of
this paper and pointing out that the embedding dimension in
Theorem~\ref{thm:sparse_l2} can be easily improved from $\bigO(d^4/\epsilon^2)$
to $\bigO(d^2/\epsilon^2)$ using the same technique.
The authors also want to thank Jelani~Nelson and Huy~Nguyen for letting us know
about their independent work on $\ell_2$ embedding.

\bibliographystyle{plain}
\bibliography{sparse_embed}

\begin{thebibliography}{10}

\bibitem{AL11}
N.~Ailon and E.~Liberty.
\newblock An almost optimal unrestricted fast {J}ohnson-{L}indenstrauss
  transform.
\newblock In {\em Proceedings of the 22nd Annual ACM-SIAM Symposium on Discrete
  Algorithms}, pages 185--191, 2011.

\bibitem{auerbach1930area}
H.~Auerbach.
\newblock {\em On the area of convex curves with conjugate diameters}.
\newblock PhD thesis, University of Lw{\'o}w, 1930.

\bibitem{AMT10}
H.~Avron, P.~Maymounkov, and S.~Toledo.
\newblock Blendenpik: Supercharging {LAPACK}'s least-squares solver.
\newblock {\em SIAM Journal on Scientific Computing}, 32:1217--1236, 2010.

\bibitem{BLM89}
J.~Bourgain, J.~Lindenstrauss, and V.~Milman.
\newblock Approximation of zonoids by zonotopes.
\newblock {\em Acta Mathematica}, 162:73--141, 1989.

\bibitem{BC05}
B.~Brinkman and M.~Charikar.
\newblock On the impossibility of dimension reduction in $\ell_1$.
\newblock {\em Journal of the ACM}, 52(5):766--788, 2005.

\bibitem{bd09}
J.~P. Brooks and J.~H. Dul\'{a}.
\newblock The {L}1-norm best-fit hyperplane problem.
\newblock {\em Applied Mathematics Letters}, 26(1):51--55, 2013.

\bibitem{chambers1976method}
J.~M. Chambers, C.~L. Mallows, and B.~W. Stuck.
\newblock A method for simulating stable random variables.
\newblock {\em Journal of the American Statistical Association},
  71(354):340--344, 1976.

\bibitem{CS02}
M.~Charikar and A.~Sahai.
\newblock Dimension reduction in the $\ell_1$ norm.
\newblock In {\em Proceedings of the 43rd Annual IEEE Symposium on Foundations
  of Computer Science}, pages 551--560, 2002.

\bibitem{Cla05}
K.~Clarkson.
\newblock Subgradient and sampling algorithms for $\ell_1$ regression.
\newblock In {\em Proceedings of the 16th Annual ACM-SIAM Symposium on Discrete
  Algorithms}, pages 257--266, 2005.

\bibitem{CDMMMW13_SODA}
K.~L. Clarkson, P.~Drineas, M.~Magdon-Ismail, M.~W. Mahoney, X.~Meng, and D.~P.
  Woodruff.
\newblock The {F}ast {C}auchy {T}ransform and faster robust linear regression.
\newblock In {\em Proceedings of the 24th Annual ACM-SIAM Symposium on Discrete
  Algorithms}, pages 466--477, 2013.

\bibitem{CW12sparse_TR}
K.~L. Clarkson and D.~P. Woodruff.
\newblock Low rank approximation and regression in input sparsity time.
\newblock Technical report.
\newblock Preprint: arXiv:1207.6365 (2012). To appear in STOC'13.

\bibitem{DDHKM09_lp_SICOMP}
A.~Dasgupta, P.~Drineas, B.~Harb, R.~Kumar, and M.~W. Mahoney.
\newblock Sampling algorithms and coresets for $\ell_p$ regression.
\newblock {\em SIAM Journal on Computing}, (38):2060--2078, 2009.

\bibitem{DKT10}
A.~Dasgupta, R.~Kumar, and T.~Sarl\'{o}s.
\newblock A sparse {J}ohnson-{L}indenstrauss transform.
\newblock In {\em Proceedings of the 42nd Annual ACM Symposium on Theory of
  Computing}, pages 341--350, 2010.

\bibitem{davidson2001local}
K.~R. Davidson and S.~J. Szarek.
\newblock Local operator theory, random matrices and {B}anach spaces.
\newblock In {\em Handbook of the Geometry of Banach Spaces}, volume~1, pages
  317--366. North Holland, 2001.

\bibitem{DMMW12_ICML}
P.~Drineas, M.~Magdon-Ismail, M.~W. Mahoney, and D.~P. Woodruff.
\newblock Fast approximation of matrix coherence and statistical leverage.
\newblock In {\em Proceedings of the 29th International Conference on Machine
  Learning}, 2012.

\bibitem{DMM06}
P.~Drineas, M.~W. Mahoney, and S.~Muthukrishnan.
\newblock Sampling algorithms for $\ell_2$ regression and applications.
\newblock In {\em Proceedings of the 17th Annual ACM-SIAM Symposium on Discrete
  Algorithms}, pages 1127--1136, 2006.

\bibitem{DMMS07_FastL2_NM10}
P.~Drineas, M.~W. Mahoney, S.~Muthukrishnan, and T.~Sarl\'{o}s.
\newblock Faster least squares approximation.
\newblock {\em Numerische Mathematik}, 117(2):219--249, 2010.

\bibitem{LN04}
J.~R. Lee and A.~Naor.
\newblock Embedding the diamond graph in ${L}_p$ and dimension reduction in
  ${L}_1$.
\newblock {\em Geometric And Functional Analysis}, 14(4):745--747, 2004.

\bibitem{levy1925calcul}
P.~L{\'e}vy.
\newblock {\em Calcul des Probabilit{\'e}s}.
\newblock Gauthier-Villars, Paris, 1925.

\bibitem{Mah-mat-rev_BOOK}
M.~W. Mahoney.
\newblock {\em Randomized Algorithms for Matrices and Data}.
\newblock Foundations and Trends in Machine Learning. NOW Publishers, Boston,
  2011.

\bibitem{CUR_PNAS}
M.~W. Mahoney and P.~Drineas.
\newblock {CUR} matrix decompositions for improved data analysis.
\newblock {\em Proc. Natl. Acad. Sci. USA}, 106:697--702, 2009.

\bibitem{maurer2003bound}
A.~Maurer.
\newblock A bound on the deviation probability for sums of non-negative random
  variables.
\newblock {\em J. Inequalities in Pure and Applied Mathematics}, 4(1), 2003.

\bibitem{MM12_TR}
X.~Meng and M.~W. Mahoney.
\newblock Low-distortion subspace embeddings in input-sparsity time and
  applications to robust linear regression.
\newblock Technical report.
\newblock Preprint: arXiv:1210.3135 (2012).

\bibitem{MSM11_TR}
X.~Meng, M.~A. Saunders, and M.~W. Mahoney.
\newblock {LSRN}: A parallel iterative solver for strongly over- or
  under-determined systems.
\newblock Technical report.
\newblock Preprint: arXiv:1109.5981 (2011).

\bibitem{mitchell2003polynomial}
J.~E. Mitchell.
\newblock Polynomial interior point cutting plane methods.
\newblock {\em Optimization Methods and Software}, 18(5):507--534, 2003.

\bibitem{nelson2012osnap}
J.~Nelson and H.~Nguyen.
\newblock {OSNAP}: Faster numerical linear algebra algorithms via sparser
  subspace embeddings.
\newblock {\em arXiv preprint arXiv:1211.1002}, 2012.

\bibitem{nolan2012stable}
J.~P. Nolan.
\newblock {\em Stable Distributions - Models for Heavy Tailed Data}.
\newblock Birkhauser, Boston, 2013.
\newblock In progress, Chapter 1 online at
  academic2.american.edu/$\sim$jpnolan.

\bibitem{RT08}
V.~Rokhlin and M.~Tygert.
\newblock A fast randomized algorithm for overdetermined linear least-squares
  regression.
\newblock {\em Proc. Natl. Acad. Sci. USA}, 105(36):13212--13217, 2008.

\bibitem{Sarlos06}
T.~Sarl\'{o}s.
\newblock Improved approximation algorithms for large matrices via random
  projections.
\newblock In {\em Proceedings of the 47th Annual IEEE Symposium on Foundations
  of Computer Science}, pages 143--152, 2006.

\bibitem{SW11}
C.~Sohler and D.~P. Woodruff.
\newblock Subspace embeddings for the $\ell_1$-norm with applications.
\newblock In {\em Proceedings of the 43rd Annual ACM Symposium on Theory of
  Computing}, pages 755--764, 2011.

\end{thebibliography}

\appendix
\section{Appendix}

\subsection{Proof of Theorem~\ref{thm:sparse_l2} (($1\pm\epsilon$)-distortion Embedding for $\ell_2$)}
\label{sec:proof_l2}

Let the $n \times d$ matrix $U$ be an orthonormal basis for the range of the $n
\times d$ matrix $A$. 
Rather than proving the theorem by establishing that
$$
(1-\epsilon) \|U z\|_2 \leq \| \Pi U z \|_2 \leq (1+\epsilon) \|U z\|_2
$$ 
holds for all $z \in \R^d$, as is essentially done in, e.g., \cite{DMM06} and
\cite{CW12sparse_TR}, we note that $U^TU=I_d$, and we directly bound the extent
to which the embedding process perturbs this product.
To do so, define
\begin{equation*}
  X = (\Pi U)^T (\Pi U) = U^T D^T S^T S D U.
\end{equation*}
That is,
\begin{equation*}
  x_{kl} = \sum_{i=1}^s \left(\sum_{j=1}^n s_{ij} d_j u_{jk}\right)\left(\sum_{j=1}^n s_{ij} d_j u_{jl} \right), \quad k, l \in \{ 1, \ldots, d \},
\end{equation*}
where $s_{ij}$ is the $(i,j)$-th element of $S$, $d_j$ is the $j$-th diagonal
element of $D$, and $u_{jk}$ is the $(j,k)$-th element of $U$.
We will use the following facts in the proof:
\begin{align*}
  \mathbf{E}[d_{j_1} d_{j_2}] &= \delta_{j_1j_2},\\
  \mathbf{E}[s_{i_1 j_1} s_{i_2 j_2}] &=
  \begin{cases}
    \frac{1}{s^2} & \text{if } j_1 \neq j_2, \\
    \frac{1}{s} & \text{if } i_1 = i_2, j_1 = j_2, \\
    0 & \text{if } i_1 \neq i_2, j_1 = j_2.
  \end{cases}
\end{align*}
We have,
\begin{align*}
  \mathbf{E}[x_{kl}] &= \sum_{i} \sum_{j_1, j_2} \mathbf{E}[ s_{ij_1} d_{j_1}
  u_{j_1k} \cdot s_{ij_2} d_{j_2} u_{j_2l} ] = \sum_{i} \sum_{j} \mathbf{E}[
  s_{ij} u_{jk} u_{jl}] = \sum_{j} u_{jk} u_{jl} = \delta_{kl},
\end{align*}
and we also have
\begin{align*}
  &\mathbf{E}[x_{kl}^2] = \mathbf{E}\left[ \left(\sum_i \left(\sum_{j} s_{ij} d_j u_{jk}\right)\left(\sum_{j} s_{ij} d_j u_{jl} \right)\right)^2 \right] \\
  &= \sum_{i_1, i_2} \mathbf{E} \left[ \left(\sum_{j} s_{i_1j} d_j u_{jk}\right)\left(\sum_{j} s_{i_1j} d_j u_{jl} \right) \left(\sum_{j} s_{i_2j} d_j u_{jk}\right)\left(\sum_{j} s_{i_2j} d_j u_{jl} \right) \right] \\
  &= \sum_{i_1, i_2} \sum_{j_1, j_2, j_3, j_4} \mathbf{E} [ s_{i_1j_1} d_{j_1} u_{j_1k} \cdot s_{i_1j_2} d_{j_2} u_{j_2 l} \cdot s_{i_2j_3} d_{j_3} u_{j_3k} \cdot s_{i_2j_4} d_{j_4} u_{j_4 l} ] \\
  &= \sum_{i_1, i_2}
  \begin{aligned}[t]
    &\left( \sum_{j} \mathbf{E} [ s_{i_1j} u_{jk} \cdot s_{i_1j} u_{j l} \cdot s_{i_2j} u_{jk} \cdot s_{i_2j} u_{j l} ]  \right.\\
    &\quad \null + \sum_{j_1 \neq j_2} \mathbf{E} [ s_{i_1j_1} u_{j_1k} \cdot s_{i_1j_1} u_{j_1 l} \cdot s_{i_2j_2} u_{j_2k} \cdot s_{i_2j_2} u_{j_2 l} ] \\
    &\quad \null + \sum_{j_1 \neq j_2 } \mathbf{E} [ s_{i_1j_1} u_{j_1k}
    \cdot s_{i_1j_2} u_{j_2 l} \cdot s_{i_2j_1} u_{j_1k} \cdot s_{i_2j_2} u_{j_2 l} ] \\
    &\quad \null + \left. \sum_{j_1 \neq j_2} \mathbf{E} [ s_{i_1j_1} u_{j_1k}
      \cdot s_{i_1j_2} u_{j_2 l} \cdot s_{i_2j_2}
      u_{j_2k} \cdot s_{i_2j_1} u_{j_1 l} ] \right) \\
  \end{aligned} \\
  &= \sum_j  u_{jk}^2 u_{jl}^2 + \sum_{j_1 \neq j_2} u_{j_1k} u_{j_1 l} u_{j_2k} u_{j_2 l} + \frac{1}{s} \sum_{j_1 \neq j_2} u_{j_1k}^2 u_{j_2 l}^2 + \frac{1}{s} \sum_{j_1 \neq j_2} u_{j_1k} u_{j_2l} u_{j_2k} u_{j_1l} \\
  &= \left( \sum_j u_{jk} u_{jl} \right)^2 + \frac{1}{s} \left( \left(\sum_j u_{jk}^2 \right)\left(\sum_j u_{jl}^2\right) + \left(\sum_j u_{jk} u_{jl}\right)^2 - 2 \sum_j u_{jk}^2 u_{jl}^2 \right) \\
  &=
  \begin{cases}
    1 + \frac{2}{s} ( 1 - \|U_{*k}\|_4^4 ) & \text{if } k = l,\\
    \frac{1}{s} ( 1 - 2 \langle U_{*k}^2, U_{*l}^2 \rangle ) & \text{if } k \neq l.\\
  \end{cases}
\end{align*}
Given these results, it is easy to obtain that 
\begin{align*}
  \mathbf{E}[\|X - I\|_F^2] = \sum_{k, l} \mathbf{E}[(x_{kl}-\delta_{kl})^2]
  = \frac{2}{s} \left( \sum_k ( 1 - \|U_{*k}\|_4^4 ) + \sum_{k < l} ( 1- 2 \langle U^2_{*k}, U^2_{*l} \rangle) \right) \leq \frac{d^2+d}{s}.
\end{align*}
For any $\delta \in (0, 1)$, set $s = (d^2+d)/(\epsilon^2 \delta)$.
Then, by Markov's inequality,
\begin{equation*}
  \mathbf{Pr}[\|X - I\|_F \geq \epsilon] = \mathbf{Pr}[\|X-I\|_F^2 \geq \epsilon^2] \leq \frac{d^2+d}{\epsilon^2 s} = \delta.
\end{equation*}
Therefore, with probability at least $1-\delta$, we have $\|X - I\|_2 \leq \|X
- I\|_F \leq \epsilon$, which implies
\begin{equation*}
  (1-\epsilon) \|U z\|_2 \leq \|\Pi U z\|_2 \leq (1+\epsilon) \|U z\|_2.
\end{equation*}

\subsection{Proof of Theorem~\ref{thm:sparse_l1} (Low-distortion Embedding for $\ell_1$)}
\label{sec:proof_l1}

We start with the following result, which establishes the existence of the
so-called Auerbach's basis of a $d$-dimensional normed vector space.
For our proof, we will only need its existence and not an algorithm to construct
it.

\begin{lemma} 
  \label{lemma:auerbach}
  (Auerbach~\cite{auerbach1930area}) Let $(\A, \|\cdot\|)$ be a $d$-dimensional
  normed vector space.
  There exists a basis $\{e_1,\ldots,e_d\}$ of $\A$, called Auerbach basis, such
  that $\|e_k\| = 1$ and $\|e^k\|^* = 1$ for $k=1,\ldots,d$, where
  $\{e^1,\ldots,e^n\}$ is a basis of $\A^*$ dual to $\{e_1,\ldots,e_n\}$.
\end{lemma}

\noindent
This Auerbach's lemma implies that a $(d, 1, 1)$-conditioned basis matrix of
$\A_1$ exists, which will be denoted by $U$ throughout the proof.
By definition, $U$'s columns are unit vectors in the $\ell_1$ norm 
(thus $|U|_1 = d$, where recall that $|\cdot|_{1}$ denotes the element-wise 
$\ell_{1}$ norm of a matrix) 
and $\|x\|_\infty \leq \|U x\|_1,\ \forall x \in \mathbb{R}^d$.
Denote by $u_j$ the $j$-th row of $U$, $j=1,\ldots,n$.
Define $v_j = \|u_j\|_1$ the $\ell_1$ leverage scores of $A$.
We have $\sum_j v_j =|U|_1 = d$.
Let $\tau > 0$ to be determined later, and define two index sets $H = \{ j \,|\,
v_j \geq \tau \}$ and $L = \{ j \,|\, v_j < \tau \}$.
It is easy to see that $|H| \leq \frac{d}{\tau}$ where $|\cdot|$ is used to
denote the size of a finite set, and $\|v^L\|_\infty \leq \tau$ where
\begin{equation*}
  v^L_j =
  \begin{cases}
    v_j, & \text{if } j \in L \\
    0, & \text{otherwise}
  \end{cases},
  \quad
  j = 1,\ldots,n.
\end{equation*}
Similarly, when an index set appears as a superscript, we mean zeroing out
elements or rows that do not belong to this index set, e.g., $v^L$ and $U^L$.
Define
\begin{equation*}
  Y = \{ y \in \mathbb{R}^n \,|\, y = U x,\ \|x\|_\infty = 1,\ x \in \mathbb{R}^d \}.
\end{equation*}
For any $y = U x \in Y$, we have $\|y\|_1 = \|U x\|_1 \geq \|x\|_\infty = 1$,
\begin{equation*}
  |y_j| = |u_j^T x| \leq \|u_j\|_1 \|x\|_\infty = v_j, \quad j = 1,\ldots,n,
\end{equation*}
and thus $\|y\|_1 \leq \|v\|_1 = d$.
Define $Y^L = \{ y \in Y \,|\, \|y^L\|_1 \geq \frac{1}{2} \|y\|_1\}$ and $Y^H =
Y \backslash Y^L$.
Given $S$, define a mapping $\phi: \{1, \ldots, n\} \to \{1, \ldots, s\}$ such
that $s_{\phi(j), j} = 1$, $j=1,\ldots,n$, and split $L$ into two subsets:
$\hat{L} = \{ j \in L \,|\, \phi(j) \in \phi(H) \}$ and $\bar{L} =
L\backslash\hat{L}$.
Consider these events:
\begin{itemize}
\item $\E_U$: $|\Pi U|_1 \leq \omega_1 d \log d$ for some $\omega_1 > 0$.

\item $\E_L$: $\|S v^L\|_\infty \leq \omega_2/(d \log d)$ for some $\omega_2 > 0$.

\item $\E_H$: $\phi(j_1) \neq \phi(j_2),\ \forall\,j_1 \neq j_2,\ j_1, j_2 \in H$.    

\item $\E_{C}$: $\min_{j \in |H|} |c_j| \geq \omega_3/(d^2 \log^2d)$
  for some $\omega_3 > 0$.

\item $\E_{\hat{L}}$: $|\Pi U^{\hat{L}}|_1 \leq \omega_4/(d^2 \log^2d)$
  for some $\omega_4 > 0$.
\end{itemize}
Recall that we set $s = \omega d^5 \log^5 d$ in Theorem~\ref{thm:sparse_l1}. 
We will show that, with $\omega$ sufficiently large and proper choices of
$\omega_1$, $\omega_2$, $\omega_3$, and $\omega_4$, the event $\E_U$
leads to an upper bound of $\|\Pi y\|_1$ for all $y \in \text{range}(A)$,
$\E_U$ and $\E_L$ lead to a lower bound of $\|\Pi y\|_1$ for
all $y \in Y^L$ with probability at least $0.9$, and $\E_H$,
$\E_{\hat{L}}$, and $\E_C$ together imply an lower bound of
$\|\Pi y\|_1$ for all $y \in Y^H$.

\begin{lemma}
  \label{lemma:upper}
  Provided $\E_U$, we have
  \begin{equation*}
    \|\Pi y\|_1 \leq \omega_1 d \log d \cdot \|y\|_1, 
    \quad \forall y \in \mathrm{range}(A).
  \end{equation*}
\end{lemma}
\begin{proof}
  For any $y \in \text{range}(A)$, we can find an $x$ such that $y = U x$.
  Then,
  \begin{equation*}
    \|\Pi y\|_1 = \|\Pi U x\|_1 \leq |\Pi U|_1 \|x\|_\infty 
    \leq |\Pi U|_1 \|U x\|_1 \leq \omega_1 d \log d \cdot \|y\|_1.
  \end{equation*}
\end{proof}
\begin{lemma}
  \label{lemma:L_1}
  Provided $\E_{L}$, for any fixed $y \in Y^L$, we have
  \begin{equation*}
    \log \Pr\left[\|\Pi y\|_1 \leq \frac{1}{4} \|y\|_1 \right] 
    \leq - \frac{d \log d}{24 \omega_2}.
  \end{equation*}
\end{lemma}

\begin{proof}
  Let $z = \Pi y$. We have,
  \begin{equation*}
    |z_i| = \left|\sum_{j} s_{ij} c_j y_j\right| 
    \simeq \left( \sum_{j} s_{ij} |y_j| \right) |\tilde{c}_i| 
    \succeq  \left( \sum_{j} s_{ij} |y^L_j| \right) |\tilde{c}_i| 
    := \tilde{\gamma}_i |\tilde{c}_i|,
  \end{equation*}
  where $\{\tilde{c}_i\}$ are independent Cauchy variables.
  Let $\tilde{\gamma} = \sum_i \tilde{\gamma}_i = \|y^L\|_1$.
  Since $|y| \leq v$, we have $\tilde{\gamma}_i \leq \|S v^L\|_\infty$.
  By Lemma~\ref{lemma:cauchy_lower},
  \begin{equation*}
    \log \Pr\left[X \leq \frac{\|y^L\|_1}{2} \right] 
    \leq  \frac{-\|y^L\|_1}{12 \|S v^L\|_\infty}.
  \end{equation*}
  By assumption $\E_L$ and $\|y^L\|_1 \geq \frac{1}{2} \|y\|_1 \geq
  \frac{1}{2}$, we obtain the result.
\end{proof}

\begin{lemma}
  \label{lemma:L_all}
  Assume both $\E_U$ and $\E_{L}$.
  If $\omega_1$ and $\omega_2$ satisfy
  \begin{equation*}
    d \log \left( 6 d ( 1 + 4 \omega_1 d \log d) \right) 
    - \frac{d \log d}{24 \omega_2} \leq \log \delta
  \end{equation*}
  for some $\delta \in (0, 1)$ regardless of $d$, then, with probability at
  least $1-\delta$, we have
  \begin{equation*}
    \|\Pi y\|_1 \geq \frac{1}{8} \|y\|_1, \quad \forall y \in Y^L.
  \end{equation*}
\end{lemma}
\begin{proof}
  Set $\epsilon = 1/(2+8 \omega_1 d \log d)$ and create an $\epsilon$-net
  $Y^L_\epsilon \subseteq Y^L$ such that for any $y \in Y^L$, we can find a
  $y_\epsilon \in Y^L_\epsilon$ such that $\|y - y_\epsilon\|_1 \leq \epsilon$.
  Since $\|y\|_1 \leq d$ for all $y \in Y^L$, there exist such an $\epsilon$-net
  with at most $(3 d/\epsilon)^d$ elements (Bourgain et al.~\cite{BLM89}).
  By Lemma~\ref{lemma:L_1}, we can apply a union bound for all the elements in
  $Y^L_\epsilon$:
  \begin{equation*}
    \Pr[\|\Pi y_\epsilon\|_1 
    \geq \frac{1}{4} \|y_\epsilon\|_1,\ \forall y_\epsilon \in Y^L_\epsilon] 
    \geq 1 - \left( \frac{3 d}{\epsilon} \right)^d e^{-\frac{d \log d}{24 \omega_2}} 
    = 1 - e^{d \log \frac{3d}{\epsilon} - \frac{d \log d}{24 \omega_2}} \geq 1 - \delta.
  \end{equation*}
  For any $y \in Y^L$, we have, noting that $y-y_\epsilon \in \text{range}(A)$,
  \begin{align*}
    \|\Pi y\|_1 &\geq \|\Pi y_\epsilon\|_1 - \|\Pi (y-y_\epsilon)\|_1 
    \geq \frac{1}{4} \|y_\epsilon\|_1 - \omega_1 d \log d \cdot \|y - y_\epsilon\|_1 \\
    &\geq \frac{1}{4} \|y\|_1 - \left( \frac{1}{4} + \omega_1 d \log d
    \right) \epsilon \geq \frac{1}{8} \|y\|_1.
  \end{align*}
  So we establish a lower bound for all $y \in Y^L$.
\end{proof}

\begin{lemma}
  \label{lemma:lower_H}
  Provided $\E_H$ and $\E_{\hat{L}}$, if $\omega_3 > 4
  \omega_4 $, we have
  \begin{equation*}
    \|\Pi y\|_1 \geq \frac{\omega_4}{d^2 \log^2 d} \|y\|_1, \quad \forall y \in Y^H.
  \end{equation*}
\end{lemma}
\begin{proof}
  For any $y = U x \in Y^H$, we have,
  \begin{align*}
    \|\Pi y\|_1 &\geq \|\Pi (y^{H} + y^{\hat{L}} )\|_1 \geq \|\Pi y^{H} \|_1 -
    \|\Pi U^{\hat{L}} x \|_1, \\
    &\geq \sum_{j \in H} |c_j| |y_j| - |\Pi U^{\hat{L}}|_1 \|x\|_\infty 
    \geq \min_{j \in H} |c_j| \|y^H\|_1 - |\Pi U^{\hat{L}}|_1 \\
    &\geq \left( \frac{\omega_3}{2 d^2 \log^2 d} - \frac{\omega_4}{d^2 \log^2 d} \right) \|y\|_1
    \geq \frac{\omega_4}{d^2 \log^2 d} \cdot \|y\|_1,
  \end{align*}
  which creates a lower bound for all $y \in Y^H$.
\end{proof}

We continue to show that, with $\omega$ sufficiently large, by setting $\tau =
\omega^{1/4} / ( d \log^2 d )$ and choosing $\omega_1$, $\omega_2$, $\omega_3$,
and $\omega_4$ properly, we have each event with probability at least $1-0.08 =
0.92$ and thus
\begin{equation*}
  \Pr[ \E_U \cap \E_L \cap \E_H 
  \cap \E_{\hat{L}} \cap \E_C ] \geq 0.6.
\end{equation*}
Moreover, the condition in Lemma~\ref{lemma:L_all} holds with $\delta = 0.1$,
and the condition in Lemma~\ref{lemma:lower_H} holds.
Therefore, $\Pi = S C$ has the desired property with probability at least $0.5$,
which would conclude the proof of Theorem~\ref{thm:sparse_l1}.

\begin{lemma}
  \label{lemma:E_U}
  With probability at least $0.92$, $\E_U$ holds with $\omega_1 = 500
  (1 + \log \omega)$.
\end{lemma}
\begin{proof}
  With $S$ fixed, we have,
  \begin{equation*}
    |\Pi U|_1 = |S C U|_1 = \sum_{k=1}^d \sum_{i=1}^s |\sum_{j=1}^n s_{ij} c_{j} u_{jk}| 
    \simeq \sum_{k=1}^d \sum_{i=1}^s \sum_{j=1}^n \left( |s_{ij} u_{jk}| \right) |\tilde{c}_{ik}|,
  \end{equation*}
  where $\{\tilde{c}_{ik}\}$ are \emph{dependent} Cauchy random variables. We have
  \begin{equation*}
    \sum_{k=1}^d \sum_{i=1}^s \sum_{j=1}^n |s_{ij} u_{jk}| = \sum_{k=1}^d \sum_{j=1}^n |u_{jk}| = |U|_1 = d.
  \end{equation*}
  Apply Lemma~\ref{lemma:cauchy_upper},
  \begin{equation*}
    \Pr[|\Pi U|_1 \geq t d \,|\, S] \leq \frac{2 \log (s d t)}{t}.
  \end{equation*}
  Setting $\omega_1 = 500 (1 + \log \omega)$ and $t = \omega_1 \log d$, we have
  \begin{equation*}
    \frac{2 \log (s d t)}{t} 
    = \frac{2 \log ( \omega \omega_1 d^6 \log^5 d )}{\omega_1 \log d} \leq 0.08.
  \end{equation*}
  We assume that $\log d \geq 1$ and $\log \omega \geq 1$.
\end{proof}

\begin{lemma}
  \label{lemma:E_L}
  For any $\delta \in (0, 0.1)$, if $s \geq d/\tau$, we have,
  \begin{equation*}
    \Pr \left[ \|S v^L \|_\infty 
      \geq \left( 1 + 2 \log \frac{d}{\delta \tau} \right) \cdot \tau \right] \leq \delta.
  \end{equation*}
\end{lemma}
\begin{proof}
  Let $X_{ij} = s_{ij} v^L_j$.
  We have $\mathbf{E}[X_{ij}] = v^L_j/s$, $\mathbf{E}[X_{ij}^2] = (v^L_j)^2/s$,
  and $0 \leq X_{ij} \leq v^L_j \leq \tau$.
  Fixed $i$, $X_{ij}$ are independent, $j=1,\ldots,n$.
  By Bernstein's inequality,
  \begin{equation*}
    \log \Pr\left[\sum_j X_{ij} 
      \geq \frac{\|v^L\|_1}{s} + t\right] \leq \frac{-t^2/2}{\|v^L\|_2^2 /s + \tau t / 3 } 
    \leq \frac{-t^2/2}{\tau ( \|v^L\|_1 / s + t / 3)} \leq \frac{- t^2/(2 \tau)}{d/s + t/3}.
  \end{equation*}
  where we use Holder's inequality: $\|v^L\|_2^2 \leq \|v^L\|_1 \|v^L\|_\infty
  \leq d \tau$.
  To obtain a union bound for all $i$ with probability $1-\delta$, we need
  \begin{align*}
    \frac{-t^2/(2 \tau)}{d/s + t/3} + \log s \leq \log \delta.
  \end{align*}
  Given $\delta < 0.1$, it suffices to choose $s = d/\tau$ and $t = 2
  \log(d/(\delta \tau)) \tau$.
  Note that $\|v^L\|_1/s \leq \|v\|_1/s = \tau$.
  We have
  \begin{equation*}
    \Pr \left[\|S v^L\|_\infty 
      \geq \left( 1 + 2 \log \frac{d}{\delta \tau} \right)\cdot \tau \right] \leq \delta.
  \end{equation*}
  Increasing $s$ will decrease the failure rate, so it holds for all $s \geq
  d/\tau$.
\end{proof}

\begin{lemma}
  With probability at least $0.92$, $\E_L$ holds with $\omega_2 = (15 +
  \log \omega)/\omega^{1/4}$.
\end{lemma}
\begin{proof}
  By Lemma~\ref{lemma:E_L}, with probability at least $0.92$, $\E_L$
  holds with
  \begin{equation*}
    \omega_2 = \frac{1+2\log \frac{\omega^{1/4} d^2 \log^2 d}{0.08}}{\omega^{1/4} \log d} 
    \leq \frac{15 + \log \omega}{\omega^{1/4}}.
  \end{equation*}
\end{proof}

\begin{lemma}
  With the above choices of $\omega_1$ and $\omega_2$, the condition in
  Lemma~\ref{lemma:L_1} holds with $\delta = 0.1$ for sufficiently large
  $\omega$.
\end{lemma}
\begin{proof}
  With $\omega_1 = 500 (1+\log \omega)$, and $\omega_2 = (15 + \log
  \omega)/\omega^{1/4}$, the first term in
  \begin{equation*}
    d \log \left( 6 d ( 1 + 4 \omega_1 d \log d ) \right) - \frac{d \log d}{24 \omega_2}
  \end{equation*}
  increases much slower than the second term as $\omega$ increases, while both
  are at the order of $d \log d$.
  Therefore, if $\omega$ is sufficiently large, the condition hold with $\delta
  = 0.1$.
\end{proof}

\begin{lemma}
  If $\omega \geq 160$, event $\E_H$ holds with probability at least
  $0.92$.
\end{lemma}
\begin{proof}
  Given $j_1, j_2 \in H$ and $j_1 \neq j_2$, let $X_{j_1j_2} = 1$ if $\phi(j_1)
  = \phi(j_2)$ and $X_{j_1j_2} = 0$ otherwise.
  It is easy to see that $\Pr[X_{j_1j_2} = 1] = \frac{1}{s}$.
  Therefore,
  \begin{equation*}
    \Pr[\E_H] \geq 1 - \sum_{j_1 < j_2} \Pr[X_{j_1j_2} = 1] 
    \geq 1 - \frac{|H|^2}{s} \geq 1 - \frac{d^2}{s \tau^2} \geq 1 - \frac{1}{\omega^{1/2}}.
  \end{equation*}
  It suffices if $\omega \geq 160$.
\end{proof}

\begin{lemma}
  With probability at least $0.92$, event $\E_C$ holds with $\omega_3 =
  1/(8 \omega^{1/4})$.
\end{lemma}
\begin{proof}
  Let $c$ be a Cauchy variable. We have
  \begin{equation*}
    \Pr[|c| \leq t] = \frac{2}{\pi} \text{tan}^{-1} t \leq \frac{2 t}{\pi}.
  \end{equation*}
  $|H|$ is at most $d/\tau = \omega^{1/4} d^2 \log^2 d$. Then
  \begin{align*}
    \Pr[\E_C] 
    &\geq 1 - |H| \cdot \Pr \left[|c| < \frac{\omega_3}{d^2 \log^2d} \right] \\
    &\geq 1 - \omega^{1/4} d^2 \log^2 d \cdot \frac{2 \omega_3}{\pi d^2 \log^2 d}.
  \end{align*}
  Therefore, $\omega_3 = 1/(8 \omega^{1/4})$ would suffice.
\end{proof}

\begin{lemma}
  With probability at least $0.92$, event $\E_{\hat{L}}$ holds with
  $\omega_4 = 25000(1+\log \omega)/\omega^{3/4}$.
  Thus with $\omega$ sufficiently large and the above choice of $\omega_3$, the
  condition in Lemma~\ref{lemma:lower_H} $\omega_3 > 4 \omega_4$ holds.
\end{lemma}
\begin{proof}
  We have,
  \begin{equation*}
    \mathbf{E}[|U^{\hat{L}}|_1] = \frac{|H|}{s} |U^{L}|_1 
    \leq \frac{\omega^{1/4} d^2 \log^2 d}{\omega d^5 \log^5 d} \cdot d 
    = \frac{1}{\omega^{3/4} d^2 \log^3 d}.
  \end{equation*}
  By Markov's inequality,
  \begin{equation*}
    \Pr\left[|U^{\hat{L}}|_1 \geq \frac{25}{\omega^{3/4} d^2 \log^3 d}\right] \leq 0.04.
  \end{equation*}
  Assume that $|U^{\hat{L}}|_1 \leq \frac{25}{\omega^{3/4} d^2 \log^3 d}$.
  Similar to the proof of Lemma~\ref{lemma:E_U}, we have
  \begin{equation*}
    |\Pi U^{\hat{L}}|_1 = \sum_{k=1}^d  \sum_{i \in \phi(H)} | \sum_j s_{ij} c_j u^{\hat{L}}_{jk} | 
    \simeq \sum_{k=1}^d \sum_{i \in \phi(H)} \left( \sum_{j} s_{ij} |u^{\hat{L}}_{jk}| \right) |\tilde{c}_{ik}|,
  \end{equation*}
  where $\{\tilde{c}_{ik}\}$ are \emph{dependent} Cauchy variables.
  Apply Lemma~\ref{lemma:cauchy_upper},
  \begin{equation*}
    \Pr[ |\Pi U^{\hat{L}}| \geq |U^{\hat{L}}| t ] \leq \frac{2 \log(|H| d t)}{t}
  \end{equation*}
  It suffices to choose $t = 1000 (1 + \log \omega) \log d$ to make the RHS less
  than $0.04$.
  So with probability at least $0.92$, we have $\E_{\hat{L}}$ holds with
  $\omega_4 = 25000(1+\log \omega)/\omega^{3/4}$.
\end{proof}

\subsection{Proof of Corollary~\ref{cor:l1reg} (Fast $\ell_1$ Regression)}
\label{sxn:pf-cor-l1reg}

By Theorem~\ref{thm:sparse_l1} and Lemma~\ref{lemma:fast_sampling}, we know that
Steps 2 and 4 of Algorithm \ref{alg:fast_l1_reg} succeed with a constant
probability. 
Conditioning on this event, we have
\begin{equation*}
  \|A \hat{x} - b\|_1 \leq \frac{1}{1-\epsilon/4} \| S A \hat{x} - S b\|_1 
  \leq \frac{1+\epsilon/4}{1 - \epsilon/4} \|S A x^* - S b\|_1 
  \leq \frac{(1+\epsilon/4)^2}{1-\epsilon/4} \|A x^* - b\|_1 
  \leq (1+\epsilon) \|A x^* - b\|_1,
\end{equation*}
where the last inequality is due to $\epsilon < 1/2$.
By Theorem~\ref{thm:sparse_l1}, Step 2 takes $\bigO(\nnz(A))$ time, and Step 3
takes $\bigO(\poly(d))$ time because $\Pi A$ has $\bigO(\poly(d)$ rows.
Then, by Lemma~\ref{lemma:fast_sampling}, Step 4 takes $\bigO(\nnz(A) \cdot \log
n)$ time, and Step 5 takes $\T_1(\epsilon/4;
\bigO(\poly(d)\log(1/\epsilon)/\epsilon^2), d)$ time.
Therefore, the total running time of Algorithm~\ref{alg:fast_l1_reg} is as
stated.

\subsection{Proof of Lemma~\ref{lemma:equiv}}
\label{sxn:pf-equiv}

First, we know that
\begin{equation*}
  \Pr[|X_p|^p \geq t] = \Pr[|X_p| 
  \geq t^{1/p}] = 2 \cdot \Pr[X_p \geq t^{1/p}].
\end{equation*}
Next, we state the following lemma, which is due to
Nolan~\cite{nolan2012stable}.
\begin{lemma}
  \label{lemma:tail}
  (Nolan~\cite[Thm.~1.12]{nolan2012stable}) Let $X \sim \D_p$ with $p
  \in [1, 2)$.
  Then as $x \to \infty$,
  \begin{equation*}
    \Pr[X > x] \sim c_p x^{-p},
  \end{equation*}
  where $c_p = \sin \frac{\pi p}{2} \cdot \Gamma(p)/\pi$.
\end{lemma}
\noindent
By Lemma~\ref{lemma:tail}, it follows that, as $t \to \infty$,
\begin{equation*}
  \Pr[|X_p|^p \geq t] \sim 2 c_p t^{-1}.
\end{equation*}
For the Cauchy distribution, we have
\begin{equation*}
  \Pr[|C| \geq t] = 1 - \frac{2}{\pi} \text{tan}^{-1} t 
  = \frac{2}{\pi} \text{tan}^{-1} \frac{1}{t} \sim \frac{2}{\pi} \cdot t^{-1}.
\end{equation*}
Hence, there exist $\alpha_p' > 0$ and $t_1 > 0$ such that for all $t > t_1$,
\begin{equation*}
  \Pr[\alpha_p' |C| \geq t] \geq \Pr[|X_p|^p \geq t].
\end{equation*}
Note that all the $p$-stable distributions with $p \in [1, 2]$ have finite and
positive density at $x = 0$.
Therefore, there exists $\alpha_p'' > 0$ such that for all $0 \leq t \leq t_1$,
\begin{equation*}
  \Pr[\alpha_p'' |C| \geq t] \geq \Pr[|X_p|^p \geq t].
\end{equation*}
Let $\alpha_p = \max \{ \alpha_p', \alpha_p'' \}$.
We get $\alpha_p |C| \succeq |X_p|^p$.
For the Gaussian distribution, we have, as $t \to \infty$,
\begin{equation*}
  \Pr[|G|^2 \geq t] \sim 2 e^{-t/2} t^{-1/2}.
\end{equation*}
which converges to zero much faster than $t^{-1}$, so we can apply similar
arguments to obtain $\beta_p$.

\subsection{Proof of Lemma~\ref{lemma:stable_upper} (Upper Tail Inequality for $p$-stable Distributions)}
\label{sxn:pf-lp-stable-upper}

Let $C_i = F_{c}^{-1}(F_p(X_i))$, $i=1,\ldots,m$, where $F_c$ is the CDF of the
standard Cauchy distribution and $F_p$ is the CDF of $\D_p$.
$C_i$ follows the standard Cauchy distribution, and, by Lemma~\ref{lemma:equiv},
we have $\alpha_p |C_i| \geq |X_i|^p$.
Therefore, for any $t \geq 1$,
\begin{align*}
  \Pr[X \geq t \alpha_p \gamma] \leq \Pr \left[ \sum_i \gamma_i |C_i| \geq t
    \gamma \right] \leq \frac{2 \log(m t)}{t}.
\end{align*}
The last inequality is from Lemma~\ref{lemma:cauchy_upper}.

\subsection{Proof of Lemma~\ref{lemma:stable_lower} (Lower Tail Inequality for $p$-stable Distributions)}
\label{sxn:pf-lp-stable-lower}

Let $G_i$ be independent random variables sampled from the standard Gaussian
distribution, $i=1, \ldots, m$.
By Lemma~\ref{lemma:equiv}, we have
\begin{align*}
  \log \Pr [ X \leq \beta_p(1-t) \gamma] \leq \log \Pr\left[ \sum_i
    \gamma_i |G_i|^2 \leq (1-t) \gamma\right].
\end{align*}
The lower tail inequality from Lemma~\ref{lemma:gaussian_lower} concludes the
proof.

\subsection{Proof of Theorem~\ref{thm:dense_lp} (Low-distortion Dense Embedding for
  $\ell_p$)}
\label{sxn:pf-dense_lp}

The proof is similar to the proof of Sohler and Woodruff~\cite[Theorem~5]{SW11},
except that the Cauchy tail inequalities are replaced by tail inequalities for
the stable distributions.
For simplicity, we omit the complete proof but show where to apply those tail
inequalities.
By Lemma~\ref{lemma:auerbach}, there exists a $(d^{1/p}, 1, p)$-conditioned
basis matrix of $\A_p$, denoted by $U$.
Thus, $|U|_p^p = d$, where recall that $|\cdot|_{p}$ denotes the element-wise 
$\ell_{p}$ norm of a matrix.
We have,
\begin{equation*}
  |\Pi U|_p^p = \sum_{k=1}^d \|\Pi u_k\|_p^p 
  = \sum_{k=1}^d \sum_{i=1}^s \left|\sum_{j=1}^n \Pi_{ij} u_{jk}\right|^p 
  \simeq \sum_{k=1}^d \sum_{i=1}^s  \|u_k\|_p^p |\tilde{X}_{ik}|^p,
\end{equation*}
where $\tilde{X}_{ik} \sim \D_p$.
Applying Lemma~\ref{lemma:stable_upper}, we get $\|\Pi U\|_p^p/s = \bigO(d \log
d)$ with a constant probability.
Define $Y = \{ U x \,|\, \|x\|_q = 1, x \in \R^d\}$.
For any fixed $y \in Y$, we have
\begin{equation*}
  \|\Pi y\|_p^p = \sum_{i=1}^s \left| \sum_{j=1}^n \Pi_{ij} y_j\right|^p 
  \simeq \sum_{i=1}^s \|y\|_p^p |\tilde{X}_i|^p,
\end{equation*}
where $\tilde{X}_i \stackrel{\text{iid}}{\sim} \D_p$.
Applying Lemma~\ref{lemma:stable_upper}, we get $\|\Pi y\|_p^p/s \leq
1/\bigO(1)$ with an exponentially small probability with respect to $s$.
By choosing $s = \omega d \log d$ with $\omega$ sufficiently large and an
$\epsilon$-net argument on $Y$, we can obtain a union lower bound of $\|\Pi
y\|_p^p$ on all the elements of $Y$ with a constant probability.
Then,
\begin{equation*}
  1/\bigO(1) \cdot \|y\|_p^p \leq \|\Pi y\|_p^p/s \leq |\Pi U|_p^p \|x\|_q^p 
  \leq \bigO(d \log d) \cdot \|U x\|_p^p = \bigO(d \log d) \|y\|_p^p, \quad y \in Y,
\end{equation*}
which gives us the desired result.

\subsection{Proof of Theorem~\ref{thm:improved-dim} (Improving the Embedding Dimension)}
\label{sxn:pf-imp-dim}

Each of Steps 1, 3, and 5 of Algorithm~\ref{alg:dim} succeeds with a constant
probability.
We can control the success rate of each by adjusting the constant factor in the
embedding dimension, such that all steps succeed with a constant probability.
Conditioning on this event, we have $\kappa_p(A R^{-1}) = 6 d$ because
\begin{align*}
  \|A R^{-1} x\|_p &\leq 2 \|\tilde{S} A R^{-1} x\|_p \leq 4 d \|x\|_2, \\
  \|A R^{-1} x\|_p &\geq \frac{2}{3} \|\tilde{S} A R^{-1} x\|_p 
  \geq \frac{2}{3} \|x\|_2, \quad \forall x \in \R^d.
\end{align*}
By Lemma~\ref{lemma:kappa_equiv}, $\bar{\kappa}_p(A R^{-1}) \leq 6 d^{1/p+1}$, and
then by Lemma~\ref{lemma:fast_sampling}, the embedding dimension of $S$ is
$\bigO( \bar{\kappa}_p^p(A R^{-1}) d^{|p/2-1|} d \log(1/\epsilon) / \epsilon^2 )
= \bigO(d^{3+p/2} \log(1/\epsilon) / \epsilon^2)$.

\end{document}